\documentclass[a4paper,12pt]{article}
\usepackage[reqno]{amsmath}
\usepackage{amsthm,amssymb,amsfonts}
\usepackage[english]{babel}
\usepackage{geometry}
\usepackage{enumitem}
\theoremstyle{plain}
\newtheorem{thm}{Theorem}[section]

\newtheorem{prop}[thm]{Proposition}
\newtheorem{lm}[thm]{Lemma}
\newtheorem{dfn}[thm]{Definition}
\newtheorem{remarks}[thm]{Remarks}

\newtheorem{example}{Example}

\def\N{\mathbb{N}}

\def\R{\mathbb{R}}

\def\P{\mathbb{P}}
\def\Tr{\mathrm{Tr}}
\def\Id{\mathrm{Id}}
\def\Im{\mathrm{Im}}
\def\Re{\mathrm{Re}}

\begin{document}
\title{On the mean-field approximation of many-boson dynamics.}
\author{Quentin~Liard \thanks{LAGA, Universit\'e de Paris 13, UMR-CNRS 7539, Campus de Villetaneuse, France.}}

\maketitle

\begin{abstract}
We show under general assumptions that the mean-field  approximation for quantum many-boson systems is accurate. Our contribution unifies and improves most of the known results.
The proof uses general properties of quantization in infinite dimensional spaces,
phase-space analysis and measure transportation techniques.
\end{abstract}

\noindent
\footnotesize {{\it Keywords}: Mean-field limit, second quantization,
  Wigner measures, continuity equation, }


\section{Introduction}
The mean-field theory provides a fair approximation  of  the dynamics and the ground state
energies of many-body quantum systems. For nearly two decades, the subject has attracted a significant  attention from the mathematical physics community and has become widely studied to this day.   Mainly  one can  distinguish  between  boson  and fermion systems with several  physically motivated scaling regimes \cite{ErScYau4}. In particular, for bosonic systems there are at least three different regimes that depend on the range of the interatomic interaction and which yield the following mean-field equations:
\begin{description}
  \item[(i)] The Gross-Pitaevskii equation  (see for instance \cite{CHPS,ErScYau1,ErScYau2,GSS,KlMa}).
  \item[(ii)] The NLS equation discussed for instance in \cite{ABGT,AmBr,ErScYau4,Hol}.
  \item[(iii)] The Hartree equation which is our main interest here.
\end{description}
The first contributions on the derivation of the Hartree dynamics  were achieved in \cite{Give1,Hep,Spo}. Several
methods and results have since been elaborated, see for instance \cite{BGGM,BGM,ElSc,ErScYau4,FGS,KnPi,RoSc}.
While the  question of accuracy of the mean-field approximation is now well-understood for the most significant examples of quantum mechanics, it has no satisfactory general mathematical answer. In fact, most of the known works in the subject deal with  a specific model or a specific choice of quantum states. Our aim here is to show that the mean-field approximation for bosonic systems is rather a general principle that depends very little on these above-mentioned specifications.

To enlighten  the discussion, we recall the meaning of  the convergence of many-body quantum dynamics towards the Hartree evolution in a concrete example and postpone the abstract framework to the previous section.  Generally, the Hamiltonian of many-boson systems has the following form,
\begin{equation}
\label{model0}
H_{N}=\sum_{i=1}^{N}-\Delta_{x_{i}}+V(x_i)+\frac{1}{N}\sum_{1\leqslant i<j\leqslant N}W(x_i-x_j)\,,\quad x_i,x_j\in \R^d\,,
\end{equation}
where $V$ is a real measurable function, $W$ is a real potential, satisfying $W(-x)=W(x),\,x \in \R^{d}$. It is in principle  meaningful to include multi-particles interactions but to keep the presentation simple we avoid to do so (see \cite{AmNi3,BQ,CP2}). Since we are dealing with bosons we assume that $H_N$ is a self-adjoint operator on the symmetric tensor product space $L_s^{2}(\R^{dN})$. Recall that $L_s^{2}(\R^{dN})$ is the space of square integrable functions which are invariant with respect to any permutations of the coordinates $(x_1,\cdots,x_n)\in \R^{dN}$ with $x_i\in\R^d$ for $i=1,\cdots,N$.
 Suppose that the many-body quantum system is prepared at an initial state  $|\Psi^{(N)}\rangle \langle \Psi^{(N)}|$, such that $\Psi^{(N)}\in L_s^{2}(\R^{dN})$, then according to the Heisenberg equation the time-evolved state is given by,
$$
\varrho_N(t):=|e^{-itH_{N}}\Psi^{(N)}\rangle \langle e^{-itH_{N}}\Psi^{(N)}|\,.
$$
The mean-field approximation  provides the first asymptotics of physical measurements  in the state $\varrho_N(t)$ when the number of particles $N$ is sufficiently large. Precisely, the approximation
deals with the following quantities,
\begin{eqnarray}
\label{mflim}
\lim_{N\to\infty}\Tr[\varrho_N(t)\, B\otimes 1^{\otimes (N-k)}]\,,
\end{eqnarray}
where $B$ is a given observable on the $k$ first particles ($k$ is kept fixed while $N\to +\infty$). It is common to express the above mean-field limit \eqref{mflim} with the language of  reduced density matrices. However, here we use the point of view of Wigner measures introduced by Ammari and Nier in  \cite{AmNi1}. In fact  one can prove that, up to extracting a subsequence, there exists a Borel probability measure $\mu_0$ on $L^{2}(\R^{d})$ such that
\begin{equation}
\label{lim0}
\lim_{N\to\infty}\Tr[\varrho_N(0) \,B\otimes 1^{\otimes (N-k)}]=\int_{L^{2}(\R^{d})} \langle z^{\otimes k} , B z^{\otimes k}\rangle_{L^{2}(\R^{dk})} \, d\mu_0(z)\,,
\end{equation}
for any compact operator $B\in\mathcal{L}(L^{2}(\R^{dk}))$, $1\leq k\leq N$.  Thus, this allows to understand the structure of the limit \eqref{mflim} at time $t=0$. Moreover, there is no loss of generality if we suppose that \eqref{lim0} holds true for the sequence of states $(\varrho_N(0))_{N\in\N}$.
Once this is observed then the accuracy of the mean-field approximation is equivalent to say that for all times $t\in\R$,
\begin{equation}
\label{cvrdm}
\lim_{N\to\infty}\Tr[\varrho_N(t) \,B\otimes 1^{\otimes (N-k)}]=\int_{L^{2}(\R^{d})} \langle z^{\otimes k} , B z^{\otimes k}\rangle_{L^{2}(\R^{dk})} \, d\mu_t(z)\,,
\end{equation}
 with $\mu_t=\Phi(t,0)_\sharp\mu_0$ is the push-forward or the image measure of $\mu_0$ by the nonlinear flow $\Phi(t,0)$ which solves the nonlinear Hartree equation,
\begin{equation}
\label{field-eq}
      i  \partial_{t}z=(-\Delta+V) z+W*|z|^{2}z\,.
\end{equation}
For more details about the relationship between Wigner measures and reduced density matrices, we refer to \cite{AmNi3,AmFaPa}. Actually, the assertion  \eqref{cvrdm} implies the convergence of reduced density matrices in the trace-class norm if the probability measure
$\mu_0$ is concentrated on the unit sphere of $L^2(\R^d)$. Notice also that the convergence \eqref{lim0} was reformulated later on as a non-commutative weak de Finetti Theorem in \cite{LNR}.

The main result of the present article is stated in Theorem \ref{thm.main}. It shows the statement \eqref{cvrdm}-\eqref{field-eq} in an abstract framework and for  arbitrary states $\varrho_N(0)$ with $-\Delta+V$  replaced by a self-adjoint operator $A$ and $W$ by a general two-body interaction.  Of course one needs some assumptions on $\varrho_N(0)$, $A$ and the interaction. But these assumptions  seem quite general and allow to recover most of the known results in the literature and to prove new ones. We give below a brief comment on the hypothesis of Theorem \ref{thm.main}. We assume that the quantum dynamics are well defined by quadratic form perturbation argument (see assumptions \eqref{A1}-\eqref{A2}) and suppose directly that the Cauchy problem is well-posed in a suitable sense given by assumption \eqref{C1}. On the other one needs two essential requirements on $\varrho_N(0)$ and  the interaction respectively. In fact, we assume that the quantum states have asymptotically finite kinetic energy at time $t=0$, i.e.:
\begin{equation}
\label{regu}
\Tr[\varrho_N(0) \, A\otimes 1^{\otimes (N-1)}]\leq C\,,
\end{equation}
uniformly in $N$.  This is a natural assumption and in some sense a minimal one if we use energy type methods to deal with the quantum and the mean-field dynamics. The second requirement concerns the two-body interaction and it is given by the abstract conditions \eqref{D1} or \eqref{D2} explained in the next section. If we stick to the example \eqref{model0} then the condition \eqref{D1} says that $A=-\Delta+V$ should have a compact resolvent and there exists a dense subspace $D\subset D(|A|^{1/2})$ such that for any $\xi\in D$,
\begin{equation*}
\begin{aligned}
\lim_{\lambda\to+\infty} ||\langle \xi|\otimes (A+1)^{-\frac{1}{2}} \,W(x_1-x_2)
(A_1+A_2+\lambda)^{-\frac{1}{2}}||_{\mathcal{L}(L_s^2(\R^{2d}), L^2(\R^d))}=0\,,\\
\lim_{\lambda\to+\infty} ||\langle \xi|\otimes (A+\lambda)^{-\frac{1}{2}} \,W(x_1-x_2)\, (A_1+A_2+1)^{-\frac{1}{2}}||_{\mathcal{L}(L_s^2(\R^{2d}),L^2(\R^d))}=0\,;
\end{aligned}
\end{equation*}
while \eqref{D2} says that there exists a dense subspace $D\subset D(|A|^{1/2})$ such that for any $\xi\in D$, the operator $
\langle \xi|\otimes (A+1)^{-\frac{1}{2}} \,W(x_1-x_2)
(A_1+A_2+1)^{-\frac{1}{2}}$ is compact from $L_s^2(\R^{2d})$ into $L^2(\R^d)$. Here $A_1,A_2$ denote the operator $A$ acting on the variable $x_1$ and $x_2$ respectively.
In particular, the mean-field approximation, with $H_N$ given in
 \eqref{model0}, is accurate in the following typical cases:
\begin{itemize}
  \item {\it{Trapped bosons}:}  If $V$ is a confining potential, for example
  $V(x)=|x|^{2}$, and  $W$ is the  potential:
$$\,\,W(x)=\frac{\lambda}{|x|^{2}}, \text{ with } |\lambda|<\frac{1}{4} \text{ and } \,d=3.$$
  \item {\it{Non-trapped bosons}:} If $V$  is infinitesimally   $-\Delta$-form bounded, for example
   $V=0$ or $V=\frac{\pm 1}{|x|}$, and $W$ is the potential:
$$\,\,W(x)=\frac{\lambda}{|x|^{\alpha}},\, \text{ with } \lambda\in\mathbb{R}\,, \; 0\leq \alpha<2 \text{ and } \,d=3.$$
\end{itemize}
 Our result emphasizes in particular the fact that the accuracy of the mean-field approximation depends very much on the criticality of the interaction and the regularity of the initial states rather than the structure of the initial states or the exact model considered. Here
 the criticality means whether the interaction is proportional  or dominated by the kinetic energy.  The method we use follows the one introduced in \cite{AmNi4} which is based on  general properties of Wick quantization in infinite dimensional spaces, Wigner measures and measure transportation techniques. We improve and simplify this method at several steps. For instance, we consider states $\varrho_N(0)$ in the symmetric tensor product $\bigvee^N\mathcal Z_0,$ where $\mathcal Z_0$ is a separable Hilbert space and avoid to work with states in the symmetric Fock space.  Moreover, the key argument related to the convergence is clarified and generalized to an abstract setting.
 The adaptation of measure transportation techniques in \cite{AGS} to handle the mean-field convergence problem was done in \cite{AmNi4} under somewhat restrictive conditions. The above techniques are recently extended in \cite{W1}  to a wider setting  which we briefly recall in Appendix B. This improvement plays a curial rule in the proof of our main result.

\bigskip
\noindent
\emph{Overview:} The main result of this article is presented in detail and illustrated with several examples in the following section. Self-adjointness and existence of the quantum dynamics
is discussed in Section \ref{se.Quantum}. The proof of our main Theorem \ref{thm.main} goes through three steps: A Duhamel-type formula in Section \ref{sub.defqj}, a convergence
argument in Section \ref{se.convarg} and a uniqueness result for a Liouville equation in Section
\ref{se.liouville}. The technical tools used along the article are explained in Appendix A and B and concern the Wick quantization, Wigner measures and transport along characteristics curves.

\section{Preliminaries and results}
\label{se.prem-resut}
In this section we introduce a general setting suitable for the study  of Hamiltonians of  many-boson systems. Then we briefly recall the notion of Wigner measures and state the main results of the present article. We will often use conventional notations.  In particular, the Banach space of bounded (resp.~compact)  operators from one Hilbert space $\mathfrak{h}_1$ into another one $\mathfrak{h}_2$ is denoted by $\mathcal{L}( \mathfrak{h}_1,\mathfrak{h}_2)$  (resp.~$\mathcal{L}^\infty( \mathfrak{h}_1,
\mathfrak{h}_2))$. If $C$ (resp.~$q$) is an operator (resp.~a quadratic form) on a Hilbert space then $D(C)$ (resp.~$Q(q)$) denotes its  domain. In particular, if
$C$ is a self-adjoint operator then $Q(C)$ denotes its form domain (i.e. the subspace $D(|C|^{\frac{1}{2}})$).

\bigskip
\noindent
\textit{General framework:}
Let $\mathcal{Z}_{0}$ be a separable Hilbert space. The $n$-fold tensor product of $\mathcal{Z}_{0}$ is denoted by
$\otimes^{n}\mathcal{Z}_{0}$.  There is an action $\sigma\in\varSigma_{n} \mapsto \Pi_{\sigma}$ of  the $n-$th symmetric group $\varSigma_{n}$ on $\otimes^{n}\mathcal Z_{0}$  verifying
\begin{equation}
\label{inter-op}
\Pi_{\sigma} (f_{1}\otimes\cdots\otimes f_{n})=f_{\sigma_{1}}\otimes\cdots\otimes f_{\sigma_{n}}\,.
\end{equation}
Hence, each $\Pi_{\sigma}$ extends to an unitary operator  on $\otimes^{n}\mathcal Z_{0}$ with the relation
$\Pi_{\sigma} \Pi_ {\sigma'}=\Pi_{\sigma\circ\sigma'}$ satisfied for any $\sigma,\sigma'\in\varSigma_{n}$.
Furthermore, the average of all  these operators  $(\Pi_\sigma)_{\sigma\in\varSigma_n}$ yields
an orthogonal projection  on $\otimes^{n}\mathcal Z_{0}$,
\begin{equation}
\label{sym}
\mathcal{S}_{n}=\frac{1}{n!}\sum_{\sigma\in
  \varSigma_{n}}\Pi_{\sigma} \,.
\end{equation}
The symmetric $n$-fold tensor product of $\mathcal{Z}_0$ is the Hilbert subspace
$$
\bigvee^{n}
\mathcal{Z}_{0}=\mathcal{S}_{n}(\mathcal{Z}_{0}^{\otimes n})\,.
$$
Consider now an  operator $A$ on $\mathcal{Z}_0$ and suppose that:\\
\\
{\bf Assumption \eqref{A1}}:
\\
\begin{align*}
\label{A1}
\tag{\bf{A1}}
A \textit{ is a non-negative and self-adjoint operator on $\mathcal{Z}_0$}.
\end{align*}
\\
For $i=1,\cdots,n$, let
$$
A_i=1^{\otimes (i-1)} \otimes
A\otimes 1^{\otimes (n-i)},
$$
where the operator $A$ in the right hand side acts on the $i^{th}$ component.  The  free Hamiltonian of a many-boson system,
\begin{equation}
\label{ham0}
H_{N}^{0}=\sum_{i=1}^{N}A_{i}\,,
\end{equation}
is  a self-adjoint non-negative operator on $\bigvee^{N}\mathcal Z_{0}$.   In order to introduce a two particles interaction in an abstract setting we consider a symmetric quadratic form
$q$ on  $Q(A_{1}+A_{2})\subset \otimes^{2}\mathcal Z_{0}$. Here $A_1+A_2$ is considered as an operator on $\otimes^{2}\mathcal Z_{0}$ and the subspace $Q(A_{1}+A_{2})$ contains in particular non-symmetric vectors. Throughout this paper we suppose:\\
\\
{\bf Assumption \eqref{A2}}:
\\
\begin{equation*}
\tag{\bf{A2}}
\label{A2}
\begin{aligned}
&\textit{$q$ is a symmetric sesquilinear form on
$Q(A_{1}+A_{2})$ satisfying  :}&\\
&\exists \,0<a<1
,\,b>0, \;\;\forall u \in Q(A_{1}+A_{2}) ,\;|q(u,u)|\leq a \langle
u,(A_{1}+A_{2})u \rangle_{\otimes^{2}\mathcal
  Z_{0}}\,+b \|u\|^{2}_{\otimes^{2}\mathcal
  Z_{0}}\,.
\end{aligned}
\end{equation*}
As a consequence of the above assumption, $q$ can be identified with a bounded  operator
$\tilde q,$ satisfying the relation:
\begin{equation}
\label{qtilde}
q(u,v)=\langle u, \tilde q \,v\rangle_{\otimes^{2}\mathcal Z_{0}}\,, \quad \forall u,v\in Q(A_{1}+A_{2})\,,
\end{equation}
and $\tilde q$ acts from the Hilbert space $Q(A_{1}+A_{2})$  equipped with the graph norm into its dual $Q'(A_{1}+A_{2})$ with respect to the inner product of $\otimes^{2}\mathcal Z_{0}$. \\
Now, we define a collection of  quadratic forms $(q_{i,j}^{(n)})_{1\leq i<j\leq n}$ by
\begin{equation}
\label{qij}
q_{i,j}^{(n)}(\varphi_{1}\otimes \varphi_{2}\otimes \cdots \otimes
\varphi_{n},\psi_{1}\otimes \psi_{2}\otimes \cdots \otimes
\psi_{n})=q(\varphi_{i}\otimes \varphi_{j},\psi_{i}\otimes
\psi_{j})\; \prod_{k\neq i,j}\langle \varphi_{k},\psi_{k}\rangle\,,
\end{equation}
for  any $\varphi_{1},\cdots,\varphi_n, \psi_1,\cdots,\psi_n$  in $Q(A)$.
By linearity all the $q_{i,j}^{(n)}$ extend to well defined quadratic forms  on the algebraic tensor product $\otimes^{alg,n}Q(A)$.
 Using the assumptions \eqref{A1}-\eqref{A2}, we prove in Lemma \ref{ses-qij} that each   $q_{i,j}^{(n)}$, $1\leq i<j\leq n$, extends uniquely to  a symmetric quadratic form on $Q(H_n^0)\subset\bigvee^n\mathcal{Z}_0$.\\
We now consider the \textit{many-boson Hamiltonian} to be the quadratic form on $Q(H_N^0)$
given by
\begin{equation}
\label{eq.hn}
H_{N}=\sum_{i=1}^{N}A_{i}+\frac{1}{N}\sum_{1\leqslant i<j\leqslant N}q_{i,j}^{(N)}=H_{N}^{0}+q_{N}.
\end{equation}
The assumptions \eqref{A1}-\eqref{A2} imply the existence of the many-boson
dynamics. In fact there exists a unique self-adjoint
operator, denoted again by $H_N$, associated to the quadratic
form \eqref{eq.hn} (see  Proposition \ref{pr.aa}).

\bigskip
\noindent
\textit{The mean-field dynamics:}
Let $(Q(A), \|\cdot\|_{Q(A)})$ be the domain form of the non-negative self-adjoint operator
$A$ equipped with the graph norm,
\begin{align*}
\|u\|_{Q(A)}^{2}=\langle u,(A+1)u\rangle,\,u \in Q(A)\,,
\end{align*}
and $Q'(A)$ its dual with respect to the inner product of $\mathcal Z_0$.
The quadratic form $q$ defines a quartic monomial,
$$
z\in Q(A)\mapsto q_0(z):=\frac{1}{2} q(z^{\otimes 2},z^{\otimes 2})\,,
$$
which is G\^ateaux differentiable on $Q(A)$. Hence, one can define the G\^ateaux derivative of $q_0$ with respect to $\bar{z}$ according to the formula:
\begin{equation}
\label{q0}
\partial_{\bar z}q_0(z)[u]=\frac{1}{2}\bar{\partial}_{\lambda}q((z+\lambda u)^{\otimes 2},(z+\lambda u)^{\otimes 2})_{|_{\lambda=0}}\,.
\end{equation}
For each $z\in Q(A)$, the map $u \mapsto \partial_{\bar z}q_0(z)[u]$ is a anti-linear continuous form on $Q(A)$ and hence $\partial_{\bar{z}}q_0(z)$ can be identified with a vector $\partial_{\bar z}q_0(z)\in Q'(A)$ by the Riesz representation theorem.\\
Consider now the following mean-field equation,
\begin{equation}
\label{field-eq2}
\left\{
    \begin{array}{l c l}
     i\partial_{t}\gamma(t)=A\gamma(t)+\partial_{\bar z}q_{0}(\gamma(t))\,, &&   \\ \\
      \gamma(0)=z_0 \in Q(A)\,.& \,&
   \end{array}
    \right.
\end{equation}
 We recall that a map $I\ni t\to \gamma(t)\in Q(A)$, where $I$ is an interval containing $0$, is a \emph{strong solution} of the Cauchy problem \eqref{field-eq2} if $\gamma\in \mathcal C(I,Q(A)) \cap \mathcal C^{1}(I,Q'(A))$ and \eqref{field-eq2} is satisfied for all $t \in I$.
Remark that \eqref{field-eq2} is    a nonlinear Hamiltonian equation  admitting  two formal conserved quantities, namely the charge $||z||_{\mathcal Z_0}$ and the classical energy:
\begin{equation}
\label{cl-enr}
h(z)=\langle z, A z\rangle+q_0(z)=\langle z, A z\rangle+
\frac{1}{2} q(z\otimes z, z\otimes z)\,.
\end{equation}
The main examples of the above mean-field equation \eqref{field-eq2} are the nonlinear Schr\"odinger and Hartree equations. For a reason that will be clear later, we are interested
only on solutions of \eqref{field-eq2} with initial data $z_0$ belonging to
$B_{\mathcal{Z}_0}(0,1)\cap Q(A)$ where $B_{\mathcal{Z}_0}(0,1)$ denotes the open unit ball of $\mathcal{Z}_0$.
 Assume that the Cauchy problem \eqref{field-eq2} is locally well-posed in the  following suitable sense.
 \\
 \\
\textbf{Assumption} (\textbf{C}):
\\
\begin{description}
  \item (i)   \textit{ There exists  an open interval $I$ containing $0$ such that for any $z_0\in B_{\mathcal{Z}_0}(0,1) \cap Q(A)$ there exists a strong solution of the Cauchy problem \eqref{field-eq2} defined on $I$.}
  \item (ii)  \textit{ For any $M>0$ there exists  $C(M)>0$ such that:}
  \begin{equation}
  \tag{{\bf C}}\label{C1}
  ||\partial_{\bar z}q_{0}(u)-\partial_{\bar z}q_{0}(v)||_{\mathcal{Z}_0} \leq C(M) \,(||u||^2_{Q(A)}+
||v||^2_{Q(A)}) \,||u-v||_{\mathcal{Z}_0}\,,
  \end{equation}
\textit{for all $u,v\in Q(A)$ such that $||u||_{\mathcal{Z}_0},||v||_{\mathcal{Z}_0}\leq M$.}
  \item (iii) \textit{Let $z_0\in B_{\mathcal{Z}_0}(0,1) \cap Q(A)$ and $z_0^{(n)}\underset{n\to\infty}{\to} z_0$ in $ Q(A)$ then for $n$ large enough the strong solutions $\gamma_n$ of \eqref{field-eq2} with $\gamma_n(0)=z_0^{(n)}$ satisfy $\gamma_n\underset{n\to\infty}{\rightarrow}\gamma$ in $\mathcal C(J, Q(A))$ for any compact interval $J\subset I$.}
\end{description}
\
\\
The assumption (i) ensures the existence of strong solutions on an interval $I$ while the assumption (ii) guaranties their uniqueness (see \cite{W1}, proof of Thm.~2.4). The last condition (iii) provides the continuous dependence of solutions of the Cauchy problem \eqref{field-eq2} with respect to  initial data. In particular, we have a well-defined continuous "local flow" map:
\begin{equation}
\label{flow}
\begin{aligned}
\Phi: I\times B_{\mathcal{Z}_0}(0,1) \cap Q(A)&\longrightarrow Q(A)\\
(t,z_0) &\longrightarrow \Phi(t,0)(z_0)=\gamma(t)\,,
\end{aligned}
\end{equation}
where $\gamma$ is the solution of \eqref{field-eq2} with the initial  datum $z_0$.

\bigskip
\noindent
\textit{The Wigner measures:}
The mean-field problem is tackled  here through the Wigner measures method elaborated in \cite{AmNi1,AmNi4}. The idea of these measures  has its roots in the finite dimensional semi-classical analysis. It allows to generalize  the notion of mean-field convergence to states that are neither coherent nor factorized. For ease of reading, we briefly recall their definition here while their main features are discussed in Appendix \ref{se.app}. A normal state $\varrho_{N}$ on $\vee^{N}\mathcal Z_{0}$  is a non-negative trace-class operator such that $\Tr [\varrho_{N}]=1$.
\begin{dfn}
\label{de.wigmeas}
Let $\{\varrho_{N}:=|\Psi^{(N)}\rangle \langle \Psi^{(N)}|\}_{N \in \N}$ be a sequence  of normal states on
$\vee^{N}\mathcal Z_{0}.$ The set $\mathcal{M}(\varrho_{N}, N \in \N)$ of
Wigner measures of $(\varrho_{N})_{N \in \N}$ is the set of Borel probability
measures $\mu$ on $\mathcal Z_{0},$ such that there exists a
subsequence $(N_{k})_{k \in \N}$ satisfying:
\begin{equation}
\label{eq.wignercv}
\forall \xi\in \mathcal Z_{0}\,,\quad
\lim_{k \to +\infty}\langle \Psi^{(N_{k})},\mathcal{W}(\sqrt{2}\pi
\xi)\Psi^{(N_{k})}\rangle =\int_{\mathcal{Z}_{0}}e^{2i\pi
\Re \langle \xi,z \rangle}\;d\mu(z)\,,
\end{equation}
where $\mathcal{W}(\sqrt{2}\pi
\xi)$ is the Weyl operator in the symmetric Fock space defined in Appendix \ref{se.app},\,Definition \eqref{weyl} with
$\varepsilon=\frac{1}{N_k}$.
\end{dfn}
The above definition extends easily to mixed states. Observe that the right hand side of \eqref{eq.wignercv} is the inverse Fourier transform of
the measure $\mu$. So Wigner measures are identified through their characteristic functions.  Moreover, it was proved in \cite[Theorem 6.2]{AmNi1} that the set $\mathcal{M}(\varrho_{N}, N \in \N)$ is non-empty and according to \cite{AmNi1,AmNi2,AmNi3,AmNi4} it is a convenient tool for the study of the mean-field approximation. In particular, it allows to understand the convergence of
reduced density matrices \eqref{cvrdm}, which are  the main analyzed quantities in other
approaches, see for instance \cite{Spo}.

\subsection{Results}
\label{se.result}
\textit{Dynamical result:}
Our main result concerns the effectiveness of the mean-field approximation for general $N$-particle states and under general assumptions \eqref{D1} or \eqref{D2}. Precisely, we show that the time-dependent Wigner measures of evolved states $\varrho_N(t):=|e^{-itH_{N}}\Psi^{(N)}\rangle \langle e^{-itH_{N}}\Psi^{(N)}|$ are the push-forward, by the flow of the mean-field equation \eqref{field-eq2}, of the initial Wigner measures of $\varrho_N(0)$. Eventually, if $\varrho_N(0)$ has only one Wigner measure then  $\varrho_N(t)$
will have also one single Wigner measure described as above. Moreover, the result is applicable to either trapped (assumption \eqref{D1}) or untrapped (assumption \eqref{D2}) systems of bosons.
\\
\\
{\bf Assumption \eqref{D1}}:
\\
\\
 \textit{
$A$ has compact resolvent and there exists a subspace $D$ dense in $Q(A)$ such that for any $\xi\in D$,
\begin{equation}
\label{D1}
\tag{\bf{D1}}
\begin{aligned}
\lim_{\lambda\to+\infty} ||\langle \xi|\otimes (A+1)^{-\frac{1}{2}} \,\mathcal{S}_2\,\tilde q
(A_1+A_2+\lambda)^{-\frac{1}{2}}||_{\mathcal{L}(\bigvee^2\mathcal{Z}_0, \mathcal{Z}_0)}=0\,,\\
\lim_{\lambda\to+\infty} ||\langle \xi|\otimes (A+\lambda)^{-\frac{1}{2}} \,\mathcal{S}_2\,\tilde q (A_1+A_2+1)^{-\frac{1}{2}}||_{\mathcal{L}(\bigvee^2\mathcal{Z}_0, \mathcal{Z}_0)}=0\,.
\end{aligned}
\end{equation}}
Actually, by Assumption \eqref{A2}, the operator
$(A_{1}+A_{2}+1)^{-\frac{1}{2}}\tilde q
(A_{1}+A_{2}+1)^{-\frac{1}{2}}$ is bounded but usually not compact in
applications. Our second main assumption is given below and it implies the two limits in \eqref{D1} (see Lemma \ref{infD2}).\\
\\
{\bf Assumption \eqref{D2}}:
\\
\\
\textit{There exists a subspace $D$ dense in $Q(A)$ such that for any $\xi\in D$,
\begin{equation}
\label{D2}
\tag{\bf{D2}}
\langle \xi|\otimes (A+1)^{-\frac{1}{2}} \,\mathcal{S}_2\,\tilde q
(A_1+A_2+1)^{-\frac{1}{2}}\in \mathcal{L}^{\infty}(\bigvee^2\mathcal{Z}_0, \mathcal{Z}_0)\,.
\end{equation}}

Consider the abstract setting explained in this section with $\mathcal
Z_{0}$ a separable Hilbert space, $A$ a one-particle self-adjoint operator, $q$ a two-body interaction and $H_N$ the many-body Hamiltonian given in \eqref{eq.hn}. Recall that \eqref{A1}-\eqref{A2} and \eqref{C1} guaranty respectively the existence  of quantum and mean-field dynamics; and remember  that $\Phi(t,0)$ denotes the flow that solves the mean-field equation \eqref{field-eq2} and that  $B_{\mathcal Z_0}(0,1)$ denotes the open unit ball of $\mathcal Z_0$.
\begin{thm}
\label{thm.main}
Assume \eqref{A1}-\eqref{A2}-\eqref{C1} and suppose that either \eqref{D1} or \eqref{D2} holds true. Let  $\{\varrho_N=|\Psi^{(N)}\rangle \langle \Psi^{(N)}|\}_{N \in \N}$ be a sequence of normal states on $\bigvee^{{N}}\mathcal Z_{0}$ admitting a unique Wigner measure $\mu_{0}$ and satisfying in addition:
\begin{equation}
\displaystyle
\label{eq.apriori}
\exists C>0, \forall N\in\mathbb{N}, \;\;\langle \Psi^{(N)},H_{N}^{0}\Psi^{(N)}\rangle\leq{C N}\,.
\end{equation}
Then for any time $t\in I$, where $I$ is provided by \eqref{C1}, the family
$$\varrho_N(t)=|e^{-itH_{N}}\Psi^{(N)}\rangle \langle e^{-itH_{N}}\Psi^{(N)}|\,, \quad N \in \N\,,
$$
has a unique Wigner measure $\mu_{t}$ which is a Borel probability measure on $B_{\mathcal{Z}_0}(0,1)\cap Q(A)$. Moreover, $\mu_{t}=\Phi(t,0)_{\sharp}\mu_{0}$ is the image measure or the push-forward of $\mu_0$ by the local flow  $\Phi(t,0)$ that solves the mean-field  equation \eqref{field-eq2} and given in \eqref{flow}.
\end{thm}

\bigskip
\noindent
Here some useful  comments on the above theorem.
\begin{remarks}\ \\
1) The result remains true if we assume that $A$ is semi-bounded from below.\\
2) It is not necessary to suppose that $\varrho_N$ admits a unique Wigner measure $\mu_0$. In general the result says that for every time $t \in I$:
$$
\mathcal{M}(\varrho_N(t), N\in\N)=\{ \Phi(t,0)_{\sharp}\mu_{0}, \mu_{0}\in \mathcal{M}(\varrho_N, N\in\N)\}\,.
$$
3) Without any essential change in the proof of Theorem \ref{thm.main}, one can suppose that $\varrho_N$ is an arbitrary sequence of non-negative trace-class operator on $\bigvee^N\mathcal{Z}_0$ satisfying:
\begin{equation*}
\exists C_1>0, \forall N\in\mathbb{N}, \;\;\Tr[\varrho_N H_{N}^{0}]\leq{C_1},\,\text{and}\,\,\exists C_2>0, \Tr[\varrho_N]\leq C_2.
\end{equation*}
\end{remarks}

\subsection{Examples}
\label{se.examples}
In this section, we provide several examples to which the general
result of Theorem \ref{thm.main} is applicable. In particular, in all these examples the  mean-field dynamics are globally well-posed in $Q(A)$, in the sense that the assumption \eqref{C1} holds true with $I=\R$.

But first we observe that the two limits in \eqref{D1} are satisfied whenever
$q$ is infinitesimally  $A_1+A_2$-form bounded. This indeed allows to handle the situation when
the interaction  is {\it{subcritical}}. When  the interaction is comparable to the kinetic energy we rely directly on \eqref{D1} which seems to be the appropriate assumption in this case.

\begin{lm}
\label{infD2}
Assume \eqref{A1}-\eqref{A2} and suppose that the quadratic form $q$ is infinitesimally $A_1+A_2$-form bounded. Then for any $\xi\in Q(A)$,
\begin{eqnarray*}
\lim_{\lambda\to+\infty} ||\langle \xi|\otimes (A+1)^{-\frac{1}{2}} \,\mathcal{S}_2\,\tilde q
(A_1+A_2+\lambda)^{-\frac{1}{2}}||_{\mathcal{L}(\bigvee^2\mathcal{Z}_0, \mathcal{Z}_0)}=0\,,\\
\lim_{\lambda\to+\infty} ||\langle \xi|\otimes (A+\lambda)^{-\frac{1}{2}} \,\mathcal{S}_2\,\tilde q (A_1+A_2+1)^{-\frac{1}{2}}||_{\mathcal{L}(\bigvee^2\mathcal{Z}_0, \mathcal{Z}_0)}=0\,.
\end{eqnarray*}
\end{lm}
\begin{proof}
Let $\Phi\in \mathcal Z_0$ and $\Psi\in \bigvee^2\mathcal Z_0$ then by Cauchy-Schwarz
inequality,
\begin{eqnarray*}
|\langle \Phi, \langle \xi|\otimes (A+1)^{-\frac{1}{2}} \,\mathcal{S}_2\,\tilde q
(A_1+A_2+\lambda)^{-\frac{1}{2}} \,\Psi\rangle|&=&
|q(\mathcal{S}_2 \xi\otimes (A+1)^{-\frac{1}{2}}\Phi, (A_1+A_2+\lambda)^{-\frac{1}{2}} \,\Psi)|\\
&\leq& |q(\mathcal{S}_2 \xi\otimes (A+1)^{-\frac{1}{2}}\Phi)|^{\frac{1}{2}}\; \; |q((A_1+A_2+\lambda)^{-\frac{1}{2}} \,\Psi)|^{\frac{1}{2}}\,,
\end{eqnarray*}
with $q(u)=q(u,u)$.  Remark that $|q(\mathcal{S}_2\xi\otimes (A+1)^{-\frac{1}{2}}\Phi)|$ is bounded thanks to \eqref{A2} and the fact that $\xi\in Q(A)$. Since $q$ is infinitesimally $A_1+A_2$-form bounded, then for any $\alpha>0$ there exists $C(\alpha)>0$ such that
\begin{eqnarray*}
|q((A_1+A_2+\lambda)^{-\frac{1}{2}} \,\Psi)|&\leq&
\alpha \langle \Psi, (A_1+A_2+\lambda)^{-1} (A_1+A_2+\frac{C(\alpha)}{\alpha}) \,\Psi\rangle\\
&\leq& \max(\alpha,\frac{C(\alpha)}{\lambda}) \,\|\Psi\|\,.
\end{eqnarray*}
This proves the first limit in \eqref{D1} when $\lambda\to\infty$. The second one follows by a
similar argument.
\end{proof}

\begin{example}[The two-body delta interaction]
\label{examp1}
Non-relativistic systems of trapped bosons  with a two-body
point interaction,
\begin{equation}
\label{delta}
H_{N}=\sum_{i=1}^{N}-\Delta_{x_{i}}+V(x_{i})+\frac{\kappa}{N}\sum_{1\leqslant
  i<j\leqslant N}\delta(x_{i}-x_{j}),\quad\,x_{i},x_{j}\in \R, \;\kappa \in \R,
\end{equation}
where $\delta$ is the Dirac distribution and $V$ is a real-valued potential which splits  into two parts $V=V_{1}+V_{2}$ such that
\begin{eqnarray*}
&&V_{1}\in L_{\text{loc}}^{1}(\R),\;V_{1}\geq 0 \,,
 \lim_{|x|\to
  +\infty}V_{1}(x)=+\infty\;,\\
&&V_{2} \;\text{is}\; -\Delta\text{-form bounded with a relative bound less than one.}
\end{eqnarray*}
This model  has been studied for instance in  \cite{ABGT,AmBr}.
The operator  $A=-\Delta+V$ is self-adjoint semi-bounded
from below and $A$ has compact resolvent according to  \cite[Theorem X19]{RS2}. The two-body interaction $q$  is given by $q(z^{\otimes
2},z^{\otimes 2})=\kappa\langle z^{\otimes
2},\delta(x_{1}-x_{2})z^{\otimes
2}\rangle=\kappa\|z\|_{L^{4}(\R)}^{4}$ and satisfies for any $u\in Q(A_1+A_2)$,
\begin{equation*}
\forall \alpha>0,\,|q(u,u)| \leq
\frac{\alpha\kappa}{2\sqrt{2}}\langle u, A_{1}+A_2\,u \rangle+\frac{\kappa}{4\alpha \sqrt{2}}\|u\|_{L^{2}(\R^{2})}^{2}.
\end{equation*}
For a detailed proof of the latter inequality see \cite[Lemma A.1]{AmBr}. Hence
\eqref{A1}-\eqref{A2} are verified and by Lemma \ref{infD2}
the assumption \eqref{D1} holds true.
The nonlinearity $\partial_{\bar z}q_0(z)=\kappa |z|^{2}z:\,Q(A) \to Q(A)$
satisfies the inequalities,
\begin{equation}
\label{eq.nlsll}
\forall z,y \in Q(A),\, \exists C:=C(\|z\|_{Q(A)},\|y\|_{Q(A)})>0,\;\||z|^{2}z-|y|^{2}y\|_{Q(A)} \leq C
\|z-y\|_{Q(A)}\,,
\end{equation}
and
\begin{equation}
\label{ineq-d=1}
\forall z,y \in Q(A),\,\||z|^{2}z-|y|^{2}y\|_{L^{2}(\R)}\leq C(\|z\|^2_{H^{1}(\R)}+ \|y\|^2_{H^{1}(\R)})\|z-y\|_{L^{2}(\R)},
\end{equation}
since the inclusion $Q(A)\subset H^{1}(\R) \subset L^{\infty}(\R)$
holds by Sobolev embedding and the fact that $Q(A)=\{u\in L^2(\R), u'\in L^2(\R),
 V_1^{\frac{1}{2}}u\in L^2(\R)\}$. Therefore the vector field  $\partial_{\bar z}q_0(z)$ is locally Lipschitz in $Q(A)$ and the \eqref{NLS} equation
\begin{equation}
\tag{{\bf{NLS}}}
\label{NLS}
\left\{
      \begin{aligned}
      i  \partial_{t}z&=-\Delta z+Vz+\kappa |z|^{2}z&\\
        z_{|t=0}&=z_{0},
         \end{aligned}
    \right.
\end{equation}
is locally well-posed in $Q(A)$ and we have a continuous dependence on the initial data . The estimate \eqref{ineq-d=1} gives (ii) in \eqref{C1}.  Furthermore, using the energy and
charge conservation one shows the global well-posedness of the
\eqref{NLS} equation and Theorem \ref{thm.main} holds true.
\end{example}

\begin{example}[Trapped bosons]
\label{examp2}
Non relativistic trapped  many-boson systems with singular two-body potentials:
\begin{equation}
\label{manybody}
H_{N}=\sum_{i=1}^{N}-\Delta_{x_{i}}+V(x_{i})+\frac{1}{N}\sum_{1\leqslant i<j
  \leqslant N}W(x_{i}-x_{j}),\quad \,x_{i},x_{j} \in \R^{d}.
\end{equation}
where $V$ is a real-valued potential which splits into two parts $V=V_{1}+V_{2}$ such that:
\begin{eqnarray*}
&&V_1\in C^{\infty}(\R^d, \R), \,V_1\geq 0,\, D^\alpha V_1\in L^{\infty}(\R^d),\; \forall
\alpha\in \N^d, |\alpha|\geq 2\,,\\
&&V_1(x)\to\infty,\, \text{ when } |x|\to\infty\,,\\
&&V_2 \in  L^{p}(\R^{d})+L^{\infty}(\R^{d}),\;p\geq 1,\; p> \frac{d}{2}\,,
\end{eqnarray*}
and  $W:\,\R^{d}\to \R$ is an even measurable function verifying:
\begin{equation}
\label{E2d}
 W \in  L^{q}(\R^{d})+L^{\infty}(\R^{d}),\;q\geq 1,\; q\geq \frac{d}{2}, (\text{and } q>1 \text{ if } d=2).
\end{equation}
By Gagliardo-Nirenberg inequality we know that $\eqref{E2d}$ implies
 that $W$ is infinitesimally $-\Delta$-form bounded. So, the assumptions \eqref{A1}-\eqref{A2} and \eqref{D1} are satisfied. Moreover, the vector field  $[\partial_{\bar z}q_0](z)=W*|z|^{2}z:\,Q(A) \to L^2(\R^d)$ satisfies  for any $z,y\in Q(A)$,
\begin{equation}
\label{est.ll}
\|W*|z|^2 z-W*|y|^2y\|_{L^2(\R^d)}\lesssim \, (\|z\|_{H^1(\R^d)}^2+ \|y\|_{H^1(\R^d)}^2)\; \|z-y\|_{L^{2}(\R^{d})}\,,
\end{equation}
by using Young, H\"{o}lder and Sobolev inequalities. So the estimate (ii) of  \eqref{C1} holds true. The global well-posedness in $Q(A)$, conservation of energy and charge of the Hartree equation
\begin{equation}
\tag{{\bf{Hartree}}}
\left\{
      \begin{aligned}
      i  \partial_{t}z&=-\Delta z+Vz+W*|z|^{2}z&\\
        z_{t=0}&=z_{0},
         \end{aligned}
    \right.
\end{equation}
are proved in \cite{Caz1} Theorem 9.2.6 and Remark 9.2.8. Observe that the assumption on $W$ are satisfied by the  Coulomb type potentials
$\frac{\lambda}{|x|^{\alpha}}$ when $\alpha<2$, $\lambda \in \R$ and $d=3$.
\end{example}

\begin{example}[Untrapped bosons]
\label{examp3}
Non-relativistic untrapped many-boson systems,
\begin{equation*}
H_{N}=\sum_{i=1}^{N}-\Delta_{x_{i}}+V(x_{i})+\frac{1}{N}\sum_{1\leqslant i<j
  \leqslant N}W(x_{i}-x_{j}),\quad \,x_{i},x_{j} \in \R^{d}.
\end{equation*}
where the potentials $V$  and $W$ satisfy the following assumptions for some $p$ and $q$,
\begin{equation}
\label{E3d}
\begin{aligned}
&V \in L^{p}(\R^{d})+L^{\infty}(\R^{d}), \,p\geq 1,\;\,p>\frac{d}{2},\\
&W \in L^{q}(\R^{d})+L_{0}^{\infty}(\R^{d}),  \,q\geq 1,\;\,q\geq\frac{d}{2}\,,(\text{and } q>1 \text{ if } d=2)\,.
\end{aligned}
\end{equation}
Here $L_{0}^{\infty}(\R^{d})$ denotes the space of bounded measurable functions going to $0$ at infinity. For instance Coulomb type potentials $\frac{\lambda}{|x|^{\alpha}}$ for $\alpha<2$, $\lambda \in \R$ and $d=3$ satisfy \eqref{E3d}.  As in the previous example \eqref{A1}-\eqref{A2} are satisfied and
\eqref{D2} is verified if we check that $(1-\Delta_x)^{-\frac{1}{2}} W(x) (1-\Delta_x)^{-\frac{1}{2}}$  is compact (see the proof of \cite[Lemma 3.10]{AmNi4}). $W$ decomposes as $W=W_{1}+W_{2}$ with $W_1\in L^{q}(\R^{d})$ and $W_2\in L_{0}^{\infty}(\R^{d})$.
We know that $W_{2}(1-\Delta)^{-\frac{1}{2}} \in
\mathcal{L}^{\infty}(L^{2}(\R^{d}))$ (see for instance
\cite[Proposition 3.21]{LoHiVo}).  Therefore we only need to check that $(1-\Delta_x)^{-\frac{1}{2}} W_{1}(x) (1-\Delta_x)^{-\frac{1}{2}}$ is compact.
Let $\chi \in \mathcal{C}_{0}^{\infty}(\R^{d})$ such that $0\leq \chi \leq 1$ and
$\chi=1$ in a neighborhood of $0$. We denote $\chi_{m}(x):=\chi(\frac{x}{m})$, for $x \in \R^{d}$ and $m \in \N^{*}$. For a given measurable function $g$ let $(g^{\delta})_{\delta>0}$ denotes
\begin{equation}
g^{\delta}=
\left\{
      \begin{aligned}
      &g, \text{ if } \,|g|<\delta&\\
      &\delta,  \text{ if } \,g\geq \delta&\\
      &-\delta,  \text{ if } \,g\leq -\delta\,. &
         \end{aligned}
    \right.
\end{equation}
Writing the decomposition
$$W_{1}=\underbrace{(\chi_{m}W_{1})^{\delta}}_{L_{0}^{\infty}(\R^{d})}
+\underbrace{W_{1}-(\chi_{m}W_{1})^{\delta}}_{L^{q}(\R^{d})}\,,$$
we observe that
$$(1-\Delta)^{-\frac{1}{2}}(\chi_{m}W_{1})^{\delta}(1-\Delta)^{-\frac{1}{2}}
\in \mathcal{L}^{\infty}(L^{2}(\R^{d})),$$
and for $\delta \to +\infty$ and $m \to +\infty$,
\begin{equation}
\label{co}
(1-\Delta)^{-\frac{1}{2}}(\chi_{m}W_{1})^{\delta}(1-\Delta)^{-\frac{1}{2}}
\underset{m\to\infty}{\longrightarrow}
(1-\Delta)^{-\frac{1}{2}}W_{1}(1-\Delta)^{-\frac{1}{2}}\,,
\end{equation}
in the norm topology. Hence \eqref{D2} holds true. The convergence \eqref{co} is justified
by the Gagliardo-Nirenberg's inequality,
\begin{equation*}
|\langle u,\big[(\chi_{m}W_{1})^{\delta}
-W_{1}\big] u\rangle|\leq C \|(\chi_{m}W_{1})^{\delta}-W_{1}\|_{L^{q}(\R^{d})}\;
\|\nabla u\|_{L^{2}(\R^d)}^{2\alpha} \,\|u \|^{2(1-\alpha)}_{L^{2}(\R^d)}\,, \qquad \alpha=\frac{d}{2q}\,.
\end{equation*}
As in Example  \ref{examp2}, the vector field $\partial_{\bar z}q_0(\cdot)$ satisfies the inequality (ii) of \eqref{C1}. The global well-posedness in $H^1(\R^d)$, conservation of energy and charge of the Hartree equation,
\begin{equation*}
\left\{
      \begin{aligned}
      i  \partial_{t}z&=-\Delta z+Vz+W*|z|^{2}z&\\
        z_{|t=0}&=z_{0},
         \end{aligned}
    \right.
\end{equation*}
hold true according to \cite{Caz1} Corollary 4.3.3 and Corollary 6.1.2.
\end{example}

\begin{example}[Non-relativistic Bosons with magnetic field]
\label{examp4}
Non-relativistic many-boson systems with an external
magnetic  field $\mathcal{A}:\R^{d} \to \R^{d}$ and  an external electric field
$V:\R^{d}\to \R$ are described by the Hamiltonian,
\begin{equation}
H_{N}=\sum_{j=1}^{N}\big[
(-i\nabla_{x_{j}}+\mathcal{A}(x_{j}))^{2}+V(x_{j})\big]+\frac{1}{N}\sum_{1\leqslant i<j \leqslant N}W(x_{i}-x_{j}),
\end{equation}
where  $W(x)$ is an even measurable function satisfying with $\mathcal A$ and $V$  the assumptions:
\begin{eqnarray*}
&&d \geq 3,\\
&&\mathcal A\in L^2_{loc}(\R^d,\R^d),\\
\label{AM1}
&&V\in L^1_{loc}(\R^d,\R),\;V_+(x)\to\infty, \text{ when } |x|\to \infty\,,\\
&&V_- \text{ is } -\Delta\text{-form bounded with relative bound less than } 1,\\
&&W\in L^q(\R^d,\R) + L^\infty(\R^d,\R), \; \nabla W\in L^{p}(\R^d)+ L^\infty(\R^d)
\text{ for some } q>\frac{d}{2},\, p\geq\frac{d}{3}.
\end{eqnarray*}
Here $V_\pm$ denotes the positive and  negative part of the potential $V$. Let $\nabla_{\mathcal{A}}:=\nabla+i\mathcal{A}$ then the quadratic form
\begin{equation*}
H_V(\mathcal{A})[f,g]:=\int_{\R^d} \overline{\nabla_{\mathcal A}f(x)}\,
 \nabla_{\mathcal A}g(x) \,dx+ \int_{\R^d} V(x) \overline{f(x)} g(x)\,dx\,,
\end{equation*}
defined on the form domain
\begin{equation*}
\mathcal{H}_{\mathcal{A},V}^{1}(\R^{d}):=\{ \varphi \in
L^{2}(\R^{d}),\,\nabla_{\mathcal{A}}\,\varphi,\,V_+^{\frac{1}{2}}\varphi \in L^{2}(\R^{d}) \},
\end{equation*}
is closed and bounded from below and hence it defines a unique semi-bounded from below self-adjoint operator denoted $H_V (\mathcal{A})$ (see \cite{AvHeSi},\cite{LeSi}).
Moreover, $\mathcal C_0^\infty(\R^d)$ is a form core for $H_V (\mathcal{A})$.
Hence \eqref{A1} is true and since $W$ satisfies  the condition \eqref{E2d}  of
Example \ref{examp2} we know that $W(x_1-x_2)$ is infinitesimally $-\Delta_{x_1}-\Delta_{x_2}$-form bounded.
Applying  \cite[Theorem 2.5]{AvHeSi} one concludes that $W(x_1-x_2)$ is infinitesimally
$H_0(\mathcal A)\otimes 1+ 1\otimes H_0(\mathcal A)$-form bounded and subsequently it is infinitesimally $H_V(\mathcal A)\otimes 1+1\otimes H_V(\mathcal A)$-form bounded. Hence \eqref{A2} is also true. Moreover, according to \cite[Theorem 2.7]{AvHeSi}
$H_V(\mathcal A)$ has a compact resolvent and so assumption  \eqref{D1} is satisfied. \\
The global well-posedness in $\mathcal{H}_{\mathcal{A},V}^{1}(\R^{d})$ of the Hartree equation
with  magnetic field
\begin{equation}
\label{eq.hartreem}
\left\{
      \begin{aligned}
      i  \partial_{t}z&=(-i\nabla+\mathcal A)^{2} z+Vz+W*|z|^{2}z&\\
        z_{|t=0}&=z_{0},
         \end{aligned}
    \right.
\end{equation}
is proved in \cite{Mi} together with energy, charge conservation and continuous dependence on initial data. The estimate (ii) of  \eqref{C1} holds true since by Young, H\"{o}lder and Sobolev inequalities,
\begin{equation}
\label{eq.estc2}
\|\big[W_{1}*\bar{z}z\big]z-\big[W_{1}*\bar{y}y\big]y\|_{L^{2}(\R^{d})}\leq \|W_{1}\|_{L^{q}(\R^{d})}(\|z\|^2_{\mathcal{H}_{\mathcal{A},V}^{1}(\R^{d})}+
\|y\|^2_{\mathcal{H}_{\mathcal{A},V}^{1}(\R^{d})})\|z-y\|_{L^{2}(\R^{d})}\,,
\end{equation}
 where $W=W_1+W_2$ with $W_1\in L^q(\R^d)$ and $W_2\in L^\infty(\R^d)$. A similar estimate  holds true for $W_2$ by Young inequality. The mean-field problem for this type of model was studied in  \cite{Lu}.
\end{example}

\begin{example}[Semi-relativistic bosons with critical interaction]
\label{examp5}
This model has been presented in \cite{ElSc,MiSc} to
describe \textit{boson stars}. Semi-relativistic systems of bosons have the many-body Hamiltonian
\begin{equation*}
H_{N}=\sum_{j=1}^{N}\sqrt{-\Delta_{x_j}+m^2}+V(x_{j})+\frac{\kappa}{N}
\sum_{1\leqslant i<j \leqslant N}\frac{1}{|x_{i}-x_{j}|},\quad x_i,x_j\in\R^3\,,
\end{equation*}
with $-\kappa_{cr}<\kappa<\kappa_{cr}$, $\kappa_{cr}^{-1}:=\displaystyle 2\lim_{\alpha\to\infty}  ||\frac{1}{|x|} (-\Delta+\alpha)^{-\frac{1}{2}}||$, $m\geq 0$ and $V$ is real-valued measurable function $V=V_1+V_2$ satisfying,
\begin{eqnarray*}
&&V_1\in L^1_{loc}(\R^3), \, V_1\geq 0,\, V_1(x)\to\infty \text{ when } |x|\to\infty\,,\\
&&V_2 \text{ is } \sqrt{-\Delta}-\text{form bounded with a relative bound less than } 1\,.
\end{eqnarray*}
The quadratic form
\begin{eqnarray*}
&&A[u,u]=\langle u, \sqrt{-\Delta+m^2}\,u\rangle +\langle u, V u\rangle\,,\\
&&Q(A)=\{ u\in L^2(\R^3), \,(-\Delta+m^2)^{\frac{1}{4}}u\in L^2(\R^3), \,V_1^{\frac{1}{2}} u\in L^2(\R^3)\}\,,
\end{eqnarray*}
is semi-bounded from below and closed. So, it defines a unique self-adjoint operator denoted by $A$. In particular assumption \eqref{A1} is verified and \eqref{A2} is satisfied thanks to a Hardy type inequality (see for  instance \cite[Proposition D.3]{AmNi4}). Hence the critical value $\kappa_{cr}$ is finite and we have the following inequality for any $z,y \in H^{1/2}(\R^{3}),$
$$
||\frac{1}{|x|}*|z|^2\, z-\frac{1}{|x|}*|y|^{2}\,y\|_{L^2(\R^3)}\leq C(||z||_{H^{1/2}(\R^3)}^2 +||y||_{H^{1/2}(\R^3)}^2) \,||z-y||_{L^2(\R^3)}\,,
$$
by the weak Young inequality, Hardy inequality and the Sobolev embedding $H^{1/2}(\R^{3}) \subset L^{3}(\R^{3}).$ Furthermore, Rellich's criterion shows that $A$ has compact resolvent.  To prove the two limits in \eqref{D1}, we use the following argument. For any $\xi,\Phi\in C_0^\infty(\R^3)$ and
$\Psi\in C_0^\infty(\R^3)$,
\begin{eqnarray}
\label{cr-eq2}
|\langle \Phi, \langle \xi|\otimes 1 \mathcal{S}_2 \frac{1}{|x-y|} \Psi\rangle|\leq
||\Phi||_{L^3(\R^3)} \, ||T\Psi||_{L^{3/2}(\R^3)} \leq ||\Phi||_{H^{1/2}(\R^3)} \, ||T\Psi||_{L^{3/2}(\R^3)}\,,
\end{eqnarray}
where $T$ is the operator given by
$$
T\Psi(y):=\int_{\R^3} \overline{\xi}(x) \,\frac{1}{|x-y|} \,\Psi(x,y)\,dx\,.
$$
Using H\"older's inequality twice with the pairs $(p,q)$, $2<q<3$, $\frac{3}{2}<p<2$ and $(4,\frac{4}{3})$,
\begin{eqnarray*}
||T\Psi(y)||_{L^{3/2}(\R^3)}^{3/2}&\leq&
\int_{\R^3} \bigg| \,|\xi|^p*\frac{1}{|\cdot|^p}\bigg|^{\frac{3}{2p}}\times
\bigg(\int_{\R^3} |\Psi(x,y)|^q \,dx\bigg)^{\frac{3}{2q}} \;dy\\
&\leq &\bigg\|\,|\xi|^p*\frac{1}{|\cdot|^p}\bigg\|_{L^{6/p}(\R^3)}^{3/2p} \;
\bigg(\int_{\R^3} ||\Psi(\cdot,y)||^2_{L^q(\R^3)} \,dy\bigg)^{\frac{3}{4}}\,.
\end{eqnarray*}
By the fractional Gagliardo-Nirenberg's inequality in \cite[Corollary 2.4]{HMOW}, we see
for $0<\alpha<1$ and  $q=\frac{6}{3-\alpha}$,
\begin{eqnarray*}
||T\Psi(y)||_{L^{3/2}(\R^3)}^{3/2}&\leq&
\bigg\|\,|\xi|^p*\frac{1}{|\cdot|^p}\bigg\|_{L^{6/p}(\R^3)}^{3/2p} \;
||\Psi(.,y)||_{L^2(\R^3)}^{2(1-\alpha)} \;\;||(-\Delta)^{\frac{1}{4}} \Psi(.,y)||_{L^2(\R^3)}^{2\alpha}\,.
\end{eqnarray*}
Therefore, using the inequality $a^{\alpha} b^{(1-\alpha)}\leq
\varepsilon a+\varepsilon^{-\frac{\alpha}{1-\alpha}} b$ for any $\varepsilon,a,b>0$, we get
\begin{eqnarray}
\label{cr-eq3}
||T\Psi(y)||_{L^{3/2}(\R^3)}^{3/2}&\leq&
\bigg\|\,|\xi|^p*\frac{1}{|\cdot|^p}\bigg\|_{L^{6/p}(\R^3)}^{3/2p} \;
\bigg(\varepsilon \,\langle \Psi, \sqrt{-\Delta_{x}} \Psi\rangle_{L^2(\R^6)}+
\varepsilon^{-\frac{\alpha}{1-\alpha}} ||\Psi||^2_{L^2(\R^6)} \bigg)\,.
\end{eqnarray}
Remark that Hardy-Littlewood-Sobolev's inequality yields
\begin{eqnarray}
\label{cr-eq1}
\bigg\|\,|\xi|^p*\frac{1}{|\cdot|^p}\bigg\|_{L^{6/p}(\R^3)}
\leq C\, || \xi||_{L^\frac{6p}{6-p}}^p<\infty.
\end{eqnarray}
So the inequalities  \eqref{cr-eq2},\eqref{cr-eq3},\eqref{cr-eq1}, provide
\begin{eqnarray*}
|\langle \Phi, \langle \xi|\otimes (A+1)^{-\frac{1}{2}} \mathcal{S}_2 \frac{1}{|x-y|} (A+\lambda)^{-\frac{1}{2}}\otimes 1\Psi\rangle|\leq C \big[\varepsilon+\frac{\varepsilon^{-\frac{\alpha}{1-\alpha}}}{\lambda}\big] \,|| \xi||_{L^\frac{6p}{6-p}}^p\;
||\Phi||_{L^2(\R^3)} \; ||\Psi||_{L^2(\R^6)}\,.
\end{eqnarray*}
This proves the first limit when $\lambda\to\infty$, the second one is
similar and it is left to the reader. \\
The global well-posedness in $Q(A)$, conservation of energy and charge of the semi-relativistic Hartree equation
\begin{equation*}
\left\{
      \begin{aligned}
      &i  \partial_{t}z=\sqrt{-\Delta+m^2}\; z+V(x) z+\frac{\kappa}{|x|}*|z|^2 z&\\
        &z_{|t=0}=z_0,
         \end{aligned}
    \right.
\end{equation*}
are proved in \cite[Theorem 4]{Lenz} for all $\kappa\geq 0$. But if $0>\kappa > -\kappa_{cr},$ then \cite[Theorem 4]{Lenz} ensures the existence and uniqueness of strong solutions for
initial data in $B_{\mathcal{Z}_0}(0,1)\cap Q(A)$. So, Theorem \ref{thm.main} applies in this case for all $\kappa$ such that $|\kappa|<\kappa_{cr}$. Remark that the upper bound on $\kappa$ is due to the use of the KLMN theorem (assumption \eqref{A2}) in order to construct the quantum dynamics. This is can be avoided if we rely on  the Friedrichs extension.
\end{example}

\section{Properties of the Quantum Dynamics}
\label{se.Quantum}
In this section we show that under the assumptions \eqref{A1}-\eqref{A2} the quadratic form  \eqref{eq.hn} defines a unique self-adjoint operator $H_N$. Thereafter,  a useful  regularity property of the related quantum dynamics is stated in Proposition \ref{pr.invhn}.

\subsection{Selfadjointness}
Remember that  the quadratic form $q$ satisfies \eqref{A2} and $q_{i,j}^{(n)}$, $q_N$ are defined respectively by \eqref{qij} and \eqref{eq.hn}.
\begin{lm}
\label{ses-qij}
Assume \eqref{A1}-\eqref{A2}. Then, for any $1\leq i<j\leq n$, $q_{i,j}^{(n)}$ extends to a  symmetric quadratic form on $Q(A_i+A_j)\subset\otimes^n\mathcal{Z}_0$. Moreover,  for any   $\Phi\in Q(A_i+A_j)$,
\begin{eqnarray}
\label{est-q0}
|q_{i,j}^{(n)}(\Phi^{(n)},\Phi^{(n)})| &\leq &  a\,\langle \Phi^{(n)},A_i+A_j\,\Phi^{(n)}\rangle +b \,\|\Phi^{(n)}\|_{\otimes^{n}\mathcal Z_{0}}^{2}\,.
\end{eqnarray}
\end{lm}
\begin{proof}
Once the estimate \eqref{est-q0} is proved for any $\Phi^{(n)}\in \otimes^{alg,n}Q(A)$, the  extension of $q_{i,j}^{(n)}$ to the domain $Q(A_i+A_j)$  is  straightforward since  $\otimes^{alg,n}Q(A)$ is a form core for $A_i+A_j$.
A simple computation yields for any $\Phi^{(n)},\Psi^{(n)} \in \otimes^{alg,n}Q(A)$,
 \begin{equation}
 \label{sym-rel}
 q_{i,j}^{(n)}(\Phi^{(n)}, \Psi^{(n)} )= q_{1,2}^{(n)}(\Pi_{(i,j)}\Phi^{(n)}, \Pi_{(i,j)}\Psi^{(n)} )\,,
 \end{equation}
where $\Pi_{(i,j)}$ is the interchange operator defined in \eqref{inter-op} with
$\sigma=(i,j)$ is the particular permutation
\begin{eqnarray*}
(i,j)=\left(
\begin{matrix}
1&2& \cdots & i &\cdots &j&\cdots &n\\
i&j& \cdots & 1 &\cdots &2&\cdots &n
\end{matrix}
\right)\,.
\end{eqnarray*}
Moreover, one remarks that
$$
\langle\Pi_{(i,j)}\Phi^{(n)}, A_1+A_2 \,\Pi_{(i,j)}\Psi^{(n)} \rangle=
\langle\Phi^{(n)}, A_i+A_j \,\Psi^{(n)} \rangle \,.
$$
Hence, it is enough to prove \eqref{est-q0} for $i=1$ and $j=2$ and $\Phi^{(n)}\in \otimes^{alg,n}Q(A)$. Let $\{e_k\}_{k\in\N}$ be an O.N.B of
$\mathcal{Z}_0$ such that $e_k\in Q(A)$ for all $k\in\N$. For $r\in\N^n$, $r=(r_1,\cdots,r_n)$, we denote
$$
e(r):=e_{r_1}\otimes\cdots\otimes e_{r_n}\in \otimes^n\mathcal{Z}_0\,.
$$
 Remark that $\{e(r)\}_{r\in\N^n}$ is an O.N.B of $\otimes^n\mathcal{Z}_0$ and for any $\Phi^{(n)}\in \otimes^{alg,n}Q(A)$ one can write $\Phi^{(n)}=\sum_{r\in\N^n} \lambda(r) e(r)$ (we may assume without  loss of generality that the sum is finite).  Hence
\begin{align*}
|q_{1,2}^{(n)}(\Phi^{(n)},\Phi^{(n)})|&=
|\sum_{r,s\in \N^n} \overline{\lambda(r)} \lambda(s) \;q_{1,2}^{(n)} \big(e(r),e(s)\big)|&\\
&\leq \left|\sum_{r_3,\cdots,r_n} \,q\left( \sum_{r_1,r_2}  \lambda(r_1,r_2,r_3,\cdots,r_n)  e_{r_1}\otimes e_{r_2}\,;  \sum_{s_1,s_2} \lambda(s_1,s_2,r_3,\cdots,r_n) e_{s_1}\otimes e_{s_2}\right)   \right|&\\
&\leq a\;\sum_{r_3,\cdots,r_n}  \big\langle \sum_{r_1,r_2}  \lambda(r_1,r_2,r_3,\cdots,r_n)  e_{r_1}\otimes e_{r_2} ; A_1+A_2 \, \sum_{s_1,s_2} \lambda(s_1,s_2,r_3,\cdots,r_n) e_{s_1}\otimes e_{s_2}\big\rangle &\\&+b\,
\sum_{r_1,r_2}  |\lambda(r_1,r_2,r_3,\cdots,r_n)|^2&\\
&\leq  a\,\langle \Phi^{(n)} , A_1+A_2 \, \Phi^{(n)}\rangle +b\, \|\Phi^{(n)}\|_{\otimes^{n}\mathcal Z_{0}}^{2}\,.
\end{align*}
The second inequality follows  using \eqref{A2}.
\end{proof}

\begin{remarks}
\label{sym-rek}
A consequence of the last proof is that for any $\Psi^{(N)} ,\Phi^{(N)} \in \vee^{alg,N}Q(A)=\mathcal{S}_N \otimes^{alg,N}Q(A)$,
$$
q_{i,j}^{(N)}(\Psi^{(N)}, \Phi^{(N)} )= q_{1,2}^{(N)}(\Psi^{(N)}, \Phi^{(N)} )\,.
$$
\end{remarks}

\begin{lm}
\label{est-qlm}
Assume \eqref{A1}-\eqref{A2}. Then $q_{N}$ extends to  a symmetric quadratic form on  $Q(H_N^0)\subset\vee^N\mathcal{Z}_0$. Moreover, for any   $\Psi^{(N)}\in Q(H_{N}^{0})$,
\begin{eqnarray}
\label{est-q}
|q_{N}(\Psi^{(N)},\Psi^{(N)})| &\leq & a\,\langle \Psi^{(N)},H_N^0\Psi^{(N)}\rangle +b N\|\Psi^{(N)}\|_{\vee^{N}\mathcal Z_{0}}^{2}\,.
\end{eqnarray}
\end{lm}
\begin{proof}
As in the previous lemma, it is enough to prove the inequality \eqref{est-q} for any $\Psi \in \vee^{alg,N}Q(A)$. Lemma  \ref{ses-qij} with Remark \ref{sym-rek} yield the estimate:
\begin{eqnarray*}
|q_{N}(\Psi^{(N)},\Psi^{(N)})| &=&  \frac{N(N-1)}{2N} | q_{1,2}^{(N)}(\Psi^{(N)},\Psi^{(N)}) | \\
&\leq& \frac{N}{2}\,\big[ a \langle \Psi^{(N)},A_1+A_2\Psi^{(N)}\rangle  +b \|\Psi^{(N)}\|_{\vee^{N}\mathcal Z_{0}}^{2}\big]\,.
\end{eqnarray*}
Using the fact that $ \langle \Psi^{(N)},A_1+A_2\Psi^{(N)}\rangle=\frac{2}{N}  \langle \Psi^{(N)}, H_N^{0}\Psi^{(N)}\rangle$, we obtain the claimed inequality.
\end{proof}
The lemma above allows to use the KLMN Theorem \cite[Theorem X.17]{RS4} since $q_N$ is a small perturbation in the sense of quadratic forms of $H_N^0$ and therefore  one  obtains the selfadjointness of $H_N$.
\begin{prop}[Self-adjoint realization of $H_{N}$]
\label{pr.aa}
Assume \eqref{A1}-\eqref{A2}, then there exists a unique self-adjoint operator $H_{N}$ with $Q(H_{N})=Q(H_{N}^{0})$ satisfying for any $\Psi^{(N)},\Phi^{(N)} \in Q(H_{N}^{0})$
\begin{equation*}
\langle \Psi^{(N)},H_{N}\Phi^{(N)}\rangle = \langle\Psi^{(N)},H_{N}^{0}\Phi^{(N)}\rangle+q_{N}(\Psi^{(N)},\Phi^{(N)}).
\end{equation*}
\end{prop}

\subsection{Invariance property}
A straightforward consequence of Proposition \ref{pr.aa} is that the
form domain $Q(H_{N}^{0})$ is invariant with respect to the dynamics
of $H_{N}$. However, we would like  to have a  quantitative uniform  bound  on
$\langle\Psi_{t}^{(N)},H_{N}^{0}\Psi_{t}^{(N)}\rangle$ for every $t
\in \R$ with 
$$
\Psi_{t}^{(N)}:=e^{-itH_N}\Psi^{(N)}\,.
$$

\begin{prop}[Propagation of states on $Q(H_{N}^{0})$]
\label{pr.invhn}
Let $\Psi^{(N)} \in Q(H_{N}^{0})$ such that $\|\Psi^{(N)}\|_{\vee^{N}\mathcal Z_{0}}=1$ and satisfying:
$$\exists C>0,\, \forall N\in\N,\,\langle \Psi^{(N)},H_{N}^{0}\Psi^{(N)}\rangle \leq C N.$$
Then there exists a constant $C_{a,b}>0$ independent of $N$ such that for any
$t \in \R$ and $N\in\N$,
$$\langle {\Psi}_{t}^{(N)},H_{N}^{0} {\Psi}_{t}^{(N)}\rangle \leq C_{a,b}N.$$
\end{prop}
\begin{proof}
Since $0<a<1$ the inequality $\pm q_{N}\leq aH_{N}^{0}+bN$ implies that $H_N^{0}\leq \frac{1}{1-a} H_N +
\frac{b}{1-a} N$ in the form sense. Let $\Psi^{(N)}\in Q(H_{N}^{0})$ then for any $t \in \R$,
\begin{eqnarray*}
\langle
\Psi_{t}^{(N)},H_{N}^{0}\Psi_{t}^{(N)}\rangle &\leq& \frac{1}{1-a} \langle \Psi_{t}^{(N)},H_{N}\Psi_{t}^{(N)}\rangle
+ \frac{b}{1-a} N \\
&\leq&
\frac{1+a}{1-a}\langle \Psi^{(N)},H^{0}_{N}\Psi^{(N)}\rangle+\frac{2b}{1-a} N \\
&\leq&
\frac{(1+a)C+2b}{1-a}  N\,.
\end{eqnarray*}
The second inequality follows using  the fact that $\langle
\Psi_{t}^{(N)},H_{N}\Psi_{t}^{(N)}\rangle=\langle
\Psi^{(N)},H_{N}\Psi^{(N)}\rangle$ and  Lemma \ref{est-qlm}.
\end{proof}

\section{Duhamel's formula}
\label{sub.defqj}
The main result provided by Theorem \ref{thm.main} is the identification of the Wigner measures
of  the time-evolved states $\varrho_{N}(t)$. According to the Definition
 \ref{de.wigmeas} of Wigner measures one needs simply to compute the limit when $N\to\infty$ of
$$
\mathcal  I_{N}(t):=\Tr[\varrho_{N}(t)\,\mathcal W(\sqrt{2}\pi\xi) ]=\langle
{\Psi}_{t}^{(N)},\mathcal W(\sqrt{2}\pi\xi){\Psi}_{t}^{(N)}\rangle\,.
$$
This task may seems quite simple but since the quantum dynamics are non trivial it is unlikely that one can compute explicitly the above limits. Therefore, it is reasonable to rely on the dynamical properties of $\mathcal  I_{N}(t)$ as for non-homogenous PDE and  write a Duhamel's formula satisfied by $\mathcal  I_{N}(t)$. The point here is  that all the possible limits of $\mathcal  I_{N}(t)$ have to satisfy a limiting integral equation. And if one can solve the latter equation then it is possible to identify the Wigner measures of $\varrho_{N}(t)$. This strategy was introduced in \cite{AmNi4} for Schr\"odinger dynamics with singular potential. Here we improve it and extend it to a more general setting.

\subsection{Commutator computation}
In order to derive the aforementioned Duhamel's formula, we differentiate the quantity
 $\mathcal I_N(t)$ with respect to time. This  roughly leads to the analysis of the commutator $[\mathcal W(\sqrt{2}\pi\xi), H_N-H_N^0]$. Since the Weyl operators do not conserve the number of particles the latter quantity has to be expanded in the symmetric Fock space.
To handle  this computation  efficiently,  we use the  Wick
quantization procedure explained in Appendix \ref{se.app} and rely
particularly in the properties of the class of symbols $\mathcal
Q_{p,q}(A)$. We suggest the reading of Appendix  \ref{se.app} before going through this subsection.\\
Recall that $\mathfrak{Q}_n=Q(H_n^0)$ is a Hilbert space equipped with the inner product
\eqref{inner}. The class of monomials $\mathcal Q_{p,q}(A)$ is defined by \eqref{def.symbq} and the energy functional satisfies:
\begin{equation*}
h(z)=\langle z, A z\rangle+\frac{1}{2} q(z^{\otimes 2},z^{\otimes 2})\in  Q_{1,1}(A)+Q_{2,2}(A)\,,
\end{equation*}
with the following relation holding for all $\Psi^{(N)},\Phi^{(N)}\in Q(H_N^0)$,
$$
\langle \Psi^{(N)}, H_N\Phi^{(N)}\rangle = \langle \Psi^{(N)},\varepsilon^{-1}
h^{Wick}\Phi^{(N)}\rangle\,, \quad \text{ when } \quad \varepsilon=\frac 1 N\,.
$$
The above identity stresses the relationship between the many-body Hamiltonian $H_N$
and the Wick quantization of the energy functional $h(z)$. It allows to exploit the general properties of Wick calculus while we deal with the dynamics of $H_N$.\\
We define the following monomial $q_s$ for any $z\in Q(A)$, $s\in\mathbb{R}$,
\begin{equation}
\label{qs}
q_s(z):=\frac{1}{2}  q\big((e^{-isA}z)^{\otimes 2}, (e^{-isA}z)^{\otimes 2}\big)=\frac{1}{2}\langle (e^{-isA}z)^{\otimes 2}, \tilde q \,(e^{-isA}z)^{\otimes 2}\rangle\,,
\end{equation}
and check that under the assumption \eqref{A2},
$$
q_s\in \mathcal Q_{2,2}(A)\,\quad \text{ with } \quad \tilde q_s=\frac{1}{2}\,e^{isA}\otimes e^{isA} \mathcal S_2\,
\tilde q \mathcal S_2\,e^{-isA}\otimes e^{-isA}\in\mathcal{L}(\mathfrak{Q}_2,\mathfrak{Q}_2')\,.
$$
A simple computation yields for any $z\in Q(A)$ and $\xi\in Q(A)$,
$$
q_s(z+i\varepsilon \pi\xi)-q_s(z)=\sum_{j=1}^{4}
\varepsilon^{j-1} q_{j}(\xi,s)\,,
$$
with the  monomials $(q_{j}(\xi,s)[z])_{j=1,2,3,4}$ defined by:
\begin{equation}
\label{mon-qjs}
\begin{aligned}
&q_{1}(\xi,s)[z]=-\pi \,\Im \,q(z_s^{\otimes 2},\mathcal{S}_2\,\xi_s \otimes z_s)\,, \quad &q_{2}(\xi,s)[z]=&-\frac{\pi^{2}}{2}\Re \,q(z_s^{\otimes 2},\xi_s^{\otimes 2})+2\pi^{2}q(\mathcal{S}_2\xi_s \otimes z_s,\mathcal{S}_2\xi_s \otimes z_s)\,,\\
&q_{3}(\xi,s)[z]=\pi^{3}\Im \,q(\xi_s^{\otimes 2},\mathcal{S}_2\,\xi_s \otimes z_s)\,,\quad & q_{4}(\xi,s)[z]=& \frac{\pi^{4}}{4} \,q(\xi_s^{\otimes 2},\xi_s^{\otimes 2}),
\end{aligned}
\end{equation}
and the notation:
\begin{equation*}
\xi_s:=e^{-isA}\xi\,, \qquad z_s:=e^{-isA}z\,.
\end{equation*}

\begin{lm}
\label{lm.symb}
Assume \eqref{A1}-\eqref{A2}, then one checks that
\begin{align*}
&q_{1}(\xi,s)[z]\in\mathcal Q_{2,1}(A)+\mathcal Q_{1,2}(A)\,, \quad &q_{2}(\xi,s)[z]\in & \mathcal Q_{2,0}(A)+\mathcal Q_{0,2}(A)+\mathcal Q_{1,1}(A)\,,\\
&q_{3}(\xi,s)[z]\in \mathcal Q_{1,0}(A)+\mathcal Q_{0,1}(A)  \,,\quad & q_{4}(\xi,s)[z]\in &
\mathcal Q_{0,0}(A) .
\end{align*}
\end{lm}
\begin{proof} This result is a straightforward consequence of Proposition \ref{c-prop-wickA} (iv). However for reader convenience we provide a direct proof.
Remark that $q_{1}(\xi,s)[z]$ is a linear combination of two conjugate monomials. So it is enough to check that $q(z^{\otimes 2},\xi \otimes z)\in\mathcal Q_{1,2}(A)$. In fact, we have
\begin{eqnarray*}
b(z)=q(z^{\otimes 2},\xi \otimes z)&=&\langle z^{\otimes 2}, \mathcal S_2 \,\tilde q \,\xi \otimes z\rangle\\
&=& \langle z^{\otimes 2}, \mathcal S_2\, \tilde q \,(|\xi\rangle\otimes 1)\, z\rangle\,.
\end{eqnarray*}
This implies that there exists a unique operator $\tilde b= \mathcal S_2 \,\tilde q \,|\xi\rangle\otimes 1$ such that for any $z\in Q(A)$,
$$
b(z)=\langle z^{\otimes 2},\tilde b \,z\rangle\,.
$$
Moreover $\tilde b\in \mathcal{L}(\mathfrak{Q}_1,\mathfrak{Q}_2')$ (here $\mathfrak{Q}_n=Q(H_n^0)$) since $\xi\in Q(A)$ and
$$
\bigg((A_1+A_2+1)^{-\frac{1}{2}} \tilde q \;(A+1)^{-\frac{1}{2}}\otimes (A+1)^{-\frac{1}{2}}\bigg)\;
|(A+1)^{\frac 1 2}\xi\rangle\otimes 1\; \;
\in \mathcal{L}(\mathcal Z_0,\otimes^2\mathcal Z_0)\,.
$$
Hence $b\in\mathcal Q_{1,2}(A)$ and $\bar b\in \mathcal Q_{2,1}(A)$ according to Proposition \eqref{c-prop-wickA} (i).
\end{proof}

\begin{prop}
\label{pr.wickcommutator}
For $\xi\in Q(A)$ and  $\varepsilon=\frac{1}{N}$, we have the following equality in the sense
of quadratic forms on $Q(H_N^0)$,
\begin{equation}
\label{pr.wickunbounded}
\frac{1}{\varepsilon} \bigg[q_s^{Wick},
\mathcal W(\sqrt{2}\pi
\xi)\bigg ]=\mathcal W(\sqrt{2}\pi\xi)\big[\sum_{j=1}^{4}\varepsilon^{j-1} q_{j}(\xi,s)^{Wick}\big],
\end{equation}
where $q_{j}(\xi,s)$, $j=1,2,3,4$, are the monomials defined in \eqref{mon-qjs} and
$q_s$ is given by \eqref{qs}.
\end{prop}
\begin{proof}
This follows by applying Proposition \ref{c-prop-wickA} (v).
\end{proof}

\subsection{Integral equation}
 Let $(\Psi^{(N)})_{N\in\N}$ be  a sequence  of normalized vectors  in $Q(H_{N}^{0})\subset\bigvee^{N}\mathcal Z_0$ satisfying the hypothesis of Theorem \ref{thm.main}. The time evolved state is
$$
\varrho_{N}(t):=|\Psi_t^{(N)}\rangle \langle \Psi_t^{(N)}| \quad \text{ where } \quad {\Psi}_{t}^{(N)}:=e^{-itH_{N}}\Psi^{(N)}\,.
$$
Actually, it is convenient to work within the interaction representation
\begin{equation}
\label{varrhotilde}
\widetilde{\varrho}_{N}(t):=|\widetilde{\Psi}_t^{(N)}\rangle \langle \widetilde{\Psi}_t^{(N)}| \quad \text{ where  } \quad \widetilde{\Psi}_{t}^{(N)}:=e^{itH_{N}^{0}}e^{-itH_{N}}\Psi^{(N)} \,.
\end{equation}
Our aim in this subsection is to write an  integral  equation (or Duhamel's formula)
satisfied by the map
\begin{equation}
\label{Jnt}
t\mapsto\mathcal  J_{N}(t):=\Tr[\widetilde{\varrho}_{N}(t)\,\mathcal W(\sqrt{2}\pi\xi) ]
=\langle
\widetilde{\Psi}_{t}^{(N)},\mathcal W(\sqrt{2}\pi\xi)\,\widetilde{\Psi}_{t}^{(N)}\rangle\,,
\end{equation}
and to put it in a convenient form in order to carry on the limit $N\to\infty$.

\begin{prop}
\label{pr.ipp}
Assume \eqref{A1}-\eqref{A2} and consider a sequence  $(\Psi^{(N)})_{N\in\N}$ of normalized
vectors in $Q(H_{N}^{0})$.
Then for any $\xi \in D(A)$ the map $t\in\mathbb{R}\mapsto \mathcal J_{N}(t)$ defined in \eqref{Jnt} is $\mathcal C^{1}$ and satisfies for $\varepsilon=\frac{1}{N}$ and all $t\in\mathbb{R}$,
\begin{equation}
\label{eq.derivation}
\mathcal J_{N}(t)=\mathcal
J_{N}(0)+i\displaystyle\int_{0}^{t} \big\langle\widetilde{\Psi}_{s}^{(N)},\mathcal
W(\sqrt{2}\pi\xi)\bigg[\sum_{j=1}^{4}\varepsilon^{j-1} \bigg(q_{j}(\xi,{s})\bigg)^{Wick}\bigg]
\widetilde{\Psi}_{s}^{(N)}\big\rangle \,ds,
\end{equation}
where ${q}_{j}(\xi,s)$, $j=1,\cdots,4$, are the monomials given  in \eqref{mon-qjs}.
\end{prop}
\begin{proof}
By Stone's theorem one can see that $\mathcal J_{N}(t)$ is continuously differentiable
since $\Psi^{(N)} \in Q(H_{N})=Q(H_{N}^{0})$. So one obtains
\begin{equation*}
i\frac{\mathrm{d}}{\mathrm{d}t}\mathcal J_{N}(t)=
\langle  \widetilde{\Psi}_{t}^{(N)},\mathcal W(\sqrt{2}\pi
\xi) e^{itH_{N}^{0}} (H_N-H_N^0) e^{-itH_{N}} \Psi^{(N)}\rangle-
\langle  e^{itH_{N}^{0}} (H_N-H_N^0) e^{-itH_{N}}\Psi^{(N)},\mathcal W(\sqrt{2}\pi
\xi)\widetilde{\Psi}_{t}^{(N)}\rangle\,.
\end{equation*}
Using the fact that $\varepsilon^{-1} q^{Wick}_{|\vee^N\mathcal Z_0}=H_N-H_N^0=q_N$ in the
sense of quadratic forms on $Q(H_N^0)$ and Proposition \ref{c-prop-wickA}, we see that
\begin{eqnarray*}
\frac{\mathrm{d}}{\mathrm{d}t}\mathcal J_{N}(t)
&=&\langle -\frac{i}{\varepsilon} e^{itH_{N}^{0}} q^{Wick} e^{-itH_{N}}\Psi^{(N)},\mathcal W(\sqrt{2}\pi
\xi)\tilde{\Psi}_{t}^{(N)}\rangle+
\langle  \tilde{\Psi}_{t}^{(N)},\mathcal W(\sqrt{2}\pi
\xi) -\frac{i}{\varepsilon} e^{itH_{N}^{0}} q^{Wick} e^{-itH_{N}} \Psi^{(N)}\rangle\\
&=&\frac{i}{\varepsilon} \langle  \widetilde{\Psi}_t^{(N)}, \bigg[q_t^{Wick},
\mathcal W(\sqrt{2}\pi
\xi)\bigg ]\widetilde{\Psi}_{t}^{(N)}\rangle\,,
\end{eqnarray*}
where $q_t(z)=\frac{1}{2}q( z_t^{\otimes^2},z_t^{\otimes^2})\in \mathcal{Q}_{2,2}(A)$.
The commutator and the duality bracket in the last equations make sense since $\mathcal W(\sqrt{2}\pi\xi)\widetilde{\Psi}_{t}^{(N)} \in
Q(d\Gamma(A)+{\bf N})$ by Proposition \ref{pr.invsew}. So, the $N^{\text{th}}$ component $\big[\mathcal W(\sqrt{2}\pi\xi)\widetilde{\Psi}_{t}^{(N)}\big]^{(N)}$ belongs to $Q(H_{N}^{0})$.
Now, we conclude by applying Proposition \ref{pr.wickcommutator}.
\end{proof}

\section{Convergence arguments}
\label{se.convarg}
We have established in the previous section an integral equation  \eqref{eq.derivation}  satisfied by  the quantity $\mathcal{J}_{N}(t)$. Here we consider its  limit when $N\to\infty$. The main steps are the analysis of $\partial_t\mathcal{J}_{N}(t)$ and the extraction
of subsequences $(N_k)_{k\in\N}$ that would lead to a convergent integral equation for all times. This is  achieved under the assumptions \eqref{D1} and \eqref{D2}.

\subsection{Convergence of $\partial_t\mathcal{J}_{N}(t)$}
\label{se.conv}
The following property is crucial for the proof of convergence.
\begin{prop}
\label{pr.compactness}
Let $\{\varrho_{N}=|\Psi^{(N)}\rangle \langle \Psi^{(N)}|\}_{N \in \N^{*}}$ be a sequence of normal states on
$\vee^{N}\mathcal Z_{0}$ such that $\mathcal M(\varrho_{N}, N \in
\N)=\{ \mu \}$ and
\begin{equation}
\label{assN}
\exists C>0, \forall N\in\N, \,\langle \Psi^{(N)},H_{N}^{0}\Psi^{(N)}\rangle \leq{CN}\,.
\end{equation}
Assume \eqref{A1}-\eqref{A2} and suppose that either \eqref{D1} or \eqref{D2} is true, then for  any $\xi \in Q(A)$ and for every $s \in \R$,
\begin{equation}
\label{eq.compactness}
\lim\limits_{\underset{ N\varepsilon=1}{N \to +\infty}}
\langle
\Psi^{(N)},\mathcal W(\sqrt{2}\pi\xi)\,[q_{1}(\xi,s)]^{Wick}\,\Psi^{(N)}\rangle=\int_{\mathcal
  Z_{0}}e^{2i\pi \Re
  \langle \xi,z\rangle}q_{1}(\xi,s)[z]\;d\mu(z)\,,
\end{equation}
where $z_{s}=e^{-isA}z,\;\xi_{s}=e^{-isA}z$ and
$q_{1}(\xi,s)[z]=-\pi\, \Im \,q(z_s^{\otimes 2},\mathcal{S}_2\,\xi_s \otimes z_s)$.
\end{prop}
\begin{proof} For simplicity we assume $s=0$ since the proof goes exactly the same
when $s\neq 0$. The following expression holds for any $\xi,z \in
Q(A)$,
\begin{equation*}
2q_{1}(\xi,0)[z]=-2\pi \,\Im \,q(z^{\otimes 2},\mathcal S_2\,\xi \otimes z)=i\pi  B_{1}(z)-i\pi B_{2}(z),
\end{equation*}
with
\begin{equation*}
B_{1}(z)=\langle \xi\otimes{z},\mathcal S_2\,\tilde{q}z^{\otimes{2}}\rangle, \quad
B_{2}(z)=\langle z^{\otimes{2}},\tilde{q} \,\mathcal S_2\,(\xi \otimes z) \rangle.
\end{equation*}
By the assumption \eqref{A2}, the two symbols $B_{1}$ and $B_{2}$ belong to $\mathcal Q_{2,1}(A)$ and $\mathcal Q_{1,2}(A)$ respectively with
\begin{equation*}
\tilde{B}_{1}=\langle \xi| \otimes 1\;\mathcal{S}_2\;\tilde{q} \;\mathcal{S}_2\in \mathcal{L}(\mathfrak{Q}_2,\mathfrak{Q}_1')\,,\qquad \tilde{B}_{2}=\mathcal{S}_2\;\tilde{q}
\;\mathcal{S}_2\;|\xi \rangle \otimes 1\,\in \mathcal{L}(\mathfrak{Q}_1,\mathfrak{Q}_2')\,,
\end{equation*}
and for any $z\in Q(A)$, $B_{1}(z)=\langle z,\tilde{B}_{1} z^{\otimes
  2}\rangle$ and $ B_{2}(z)=\langle z^{\otimes
  2},\tilde{B}_{2}z\rangle$ with the property $\overline{B_1(z)}=B_{2}(z)$. \\
We will use an approximation argument.
Let $\chi \in \mathcal C_{0}^{\infty}(\R)$ such that $\chi(x)=1$ if $|x|\leq{1},$ $\chi(x)=0$ if $|x|\geq{2}$ and $0\leq{\chi}\leq{1}$. We denote for $m \in \N^*$,
$\chi_{m}(x)=\chi(\frac{x}{m})$ and $H_1^0=A$, $H_2^0=A_1+A_2$ and set
\begin{equation*}
\tilde{B}_{1,m}:=\chi_{m}(H_1^0) \,\tilde{B}_{1} \,\chi_{m}(H_2^0)\in
\mathcal{L}(\vee^2\mathcal{Z}_0, \mathcal{Z}_0) \,, \qquad
\tilde{B}_{2,m}:= \chi_{m}(H_2^0) \,
\tilde{B}_{2}\,  \chi_{m}(H_1^0)\in \mathcal{L}( \mathcal{Z}_0,\vee^2\mathcal{Z}_0) \,,
\end{equation*}
and
\begin{equation*}
B_{1,m}(z)=\langle z,\tilde{B}_{1,m} z^{\otimes
  2}\rangle\,,\qquad B_{2,m}(z)=\langle z^{\otimes
  2},\tilde{B}_{2,m}z\rangle\,.
\end{equation*}
Since \eqref{D1} says that $A$ has compact resolvent and both operators $(H_1^0+1)^{-\frac{1}{2}} \tilde{B_{1}} (H_2^0+1)^{-\frac{1}{2}}$ and $(H_2^0+1)^{-\frac{1}{2}} \tilde{B_{2}}(H_1^0+1)^{-\frac{1}{2}} $ are either compact or bounded, we see that $B_{j,m}$ are  compact operators once we assume \eqref{D1} or \eqref{D2}.  We now write the following inequalities for $j=1,2$,
\begin{align}
\label{eq.convergence}
|\langle \Psi^{(N)},\mathcal W(\sqrt{2}\pi\xi) \,B_{j}^{Wick}\,\Psi^{(N)}\rangle-\mu(e^{2i\pi \Re \langle\xi,z\rangle}B_{j}(z))|\leq{\mathcal{A}^{(m)}_{j}+\mathcal{B}^{(m)}_{j}+\mathcal{C}^{(m)}_{j}},
\end{align}
where $$\mathcal{A}^{(m)}_{j}=|\langle \Psi^{(N)},\mathcal W(\sqrt{2}\pi\xi)[B_{j}-B_{j,m}]^{Wick}\Psi^{(N)}\rangle|,$$
$$\mathcal{B}^{(m)}_{j}=|\langle\Psi^{(N)},\mathcal W(\sqrt{2}\pi\xi) \,B_{j,m}^{Wick}\Psi^{(N)}\rangle-\mu(e^{2i\pi \Re \langle\xi,z\rangle} B_{j,m}(z))|,$$
and
$$
\mathcal{C}^{(m)}_{j}=|\mu(e^{2i\pi\Re \langle \xi,z \rangle} B_{j,m}(z))-\mu(e^{2i\pi\Re\langle\xi,z\rangle} B_{j}(z))|.
$$
To prove the limit \eqref{eq.compactness}, we show that all the terms
$\mathcal{A}^{(m)}_{j},\mathcal{B}^{(m)}_{j}, \mathcal{C}^{(m)}_{j}$ can be made
arbitrary small for all $N$ larger enough by choosing a convenient $m\in\N$.

\medskip
\noindent
{\bf The term $\mathcal{C}^{(m)}_{j}$}:
By dominated convergence theorem the quantity  $\mathcal{C}^{(m)}_{j}$ tends to $0$
when $m\to\infty$ for $j=1,2$. In fact  $B_{j,m}(z)$ converges to $B_{j}(z)$ for all $z \in Q(A)$ since $s-\lim \chi_{m}(H_j^0)= \Id$. Moreover, we have for some $C'>0$  and any $z\in Q(A)$,
\begin{equation}
\label{bound1}
|B_{j,m}(z)|\leq
C'\|\xi\|_{Q(A)} \,\|z\|^{2}_{Q(A)} \,\|z\|_{\mathcal Z_{0}}\,,
\end{equation}
since $B_{j,m}$ are in $\mathcal Q_{1,2}(A)$ or $\mathcal Q_{2,1}(A)$ and by Proposition \ref{pr.transportapriori} we get the a priori estimate:
\begin{equation}
\label{bound2}
\int_{\mathcal{Z}_{0}} \|z\|^{2}_{Q(A)} \|z\|_{\mathcal Z_{0}}\;d\mu(z)\leq
C\,.
\end{equation}

\bigskip
\noindent
{\bf The term $\mathcal{B}^{(m)}_{j}$}:
Since $\tilde B_{j,m}$ are compact operators for $j=1,2$ and any $m\in\N^*$, the quantity $\mathcal{B}^{(m)}_{j}\to 0$ when $N \to\infty$ owing to result proved in \cite[Theorem 6.13]{AmNi1} and \cite[Corollary 6.14]{AmNi1}.

\bigskip
\noindent
{\bf The term $\mathcal{A}^{(m)}_{j}$}:
We consider only  $j=1$ since the case $j=2$ is quite similar. We write for any $z\in Q(A)$,
\begin{equation*}
B_1(z)-B_{1,m}(z)=\langle z, (1-\chi_m(H_1^0)) \tilde B_1 \, z^{\otimes 2}\rangle+
\langle z, \chi_m(H_1^0) \tilde B_1 (1-\chi_m(H_2^0))\,z^{\otimes 2}\rangle=:\mathcal U_1(z)
+\mathcal U_2(z)\,,
\end{equation*}
 and check  that $\mathcal U_1,\mathcal U_2\in \mathcal
 Q_{2,1}(A)$. Let  $\Phi^{(N-1)}=[\mathcal W(\sqrt{2}\pi\xi)\Psi^{(N)}]^{(N-1)}$ be the $(N-1)^{\text{th}}$
component of the vector $\mathcal W(\sqrt{2}\pi\xi)\Psi^{(N)}$ in the symmetric Fock space $\Gamma_s(\mathcal Z_0)$.~By Proposition \ref{pr.invsew} we see that $\Phi^{(N-1)}\in Q(H_{N-1}^0)$. So, one obtains
\begin{equation*}
\mathcal{A}^{(m)}_{1}=\underset{(1)}{\underbrace{\langle \Phi^{(N-1)},\mathcal U_1^{Wick} \,
\Psi^{(N)}\rangle}} +\underset{(2)}{\underbrace{\langle \Phi^{(N-1)},\mathcal U_2^{Wick} \,
\Psi^{(N)}\rangle}}.
\end{equation*}
Now estimate each term. Let denote $\overline{\chi}_m=1-\chi_m$ then for $\lambda>0$  and $\varepsilon=\frac{1}{N}$,
\begin{eqnarray*}
\bigg|(1)\bigg|&=&\bigg| \langle \Phi^{(N-1)}, \varepsilon^{3/2} \sqrt{N (N-1)^2}
\,\mathcal S_{N-1} \,\overline{\chi}_m(H^0_1)\tilde B_1 \otimes 1^{(N-2)} \,
\Psi^{(N)}\rangle\bigg| \\
 &\leq&
 \bigg| \langle \overline{\chi}_m(H^0_1)\otimes 1^{(N-2)}\Phi^{(N-1)},
\,\tilde B_1 \otimes 1^{(N-2)} \,
\Psi^{(N)}\rangle\bigg|\\
&\leq& \alpha(\lambda)
 \bigg\|(H_1^0+\lambda)^{1/2} \overline{\chi}_m(H^0_1)\otimes 1^{(N-2)}\Phi^{(N-1)}\bigg\|
 \;\;  \bigg\| (H_2^0+1)^{1/2} \otimes  1^{(N-2)}\Psi^{(N)}\bigg\|\,,
\end{eqnarray*}
where
$$
\alpha(\lambda)=\bigg\| (H_1^0+\lambda)^{-1/2} \tilde B_1 (H_2^0+1)^{-1/2} \bigg\|_{\mathcal{L}(\bigvee^2\mathcal Z_0,\mathcal Z_0)}\to 0, \quad \text{ when } \lambda\to\infty\,.
$$
Remark that the spectral theorem yields,
\begin{equation*}
\forall m \in \N^{*},\;\|\overline{\chi}_{m}(A)\;(A+1)^{-\frac{1}{2}}\|_{\mathcal L(\mathcal Z_{0})}^{2}\leq \frac{1}{m}.
\end{equation*}
So using the assumption \eqref{assN}, the symmetry of $\Phi^{(N-1)}$ and Proposition
\ref{pr.invsew}, one obtains
$$
\bigg\|(H_1^0+\lambda)^{1/2} \overline{\chi}_m(H^0_1)\otimes 1^{(N-2)}\Phi^{(N-1)}\bigg\|
\leq C_1 \sqrt{1+\frac{\lambda}{m}}\,,
$$
form some $C_1>0$ independent of $N$. Hence $|(1)|\lesssim \alpha(\lambda) \sqrt{1+\frac{\lambda}{m}}$ and if we choose $\lambda=m$ we see that $|(1)|\to 0$ when
$m\to\infty$. \\
Similar  computation yields for $\lambda$ large enough
\begin{eqnarray*}
\bigg|(2)\bigg|  &\leq& \beta(\lambda)
 \bigg\|(H_1^0+\lambda)^{1/2}\otimes 1^{(N-2)}\Phi^{(N-1)}\bigg\|
 \;\;  \bigg\|  \overline{\chi}_m(H^0_2)(H_2^0+\lambda)^{1/2} \otimes  1^{(N-2)}\Psi^{(N)}\bigg\|\,,
\end{eqnarray*}
where
$$
\beta(\lambda)=\bigg\| (H_1^0+1)^{-1/2} \tilde B_1 (H_2^0+\lambda)^{-1/2} \bigg\|_{\mathcal{L}(\vee^2\mathcal Z_0,\mathcal Z_0)}\to 0, \quad \text{ when } \lambda\to\infty\,.
$$
So by the same argument above we conclude that $|(2)|\lesssim \beta(\lambda) \sqrt{1+\frac{\lambda}{m}}$ and if we choose again  $\lambda=m$ we get $|(2)|\to 0$ when
$m\to\infty$. \\
This proves the claimed limit \eqref{eq.compactness} for any $\xi \in D\subset \mathcal Z_0$.
So we extend this result to any $\xi\in Q(A)$ by an approximation argument.
In fact take for any $\xi\in Q(A)$ a sequence $(\xi_m)_{m\in\N}$ such that $ \xi_m\to \xi$ in $Q(A)$. Write
\begin{eqnarray*}
\left|\langle\Psi^{(N)},\mathcal W(\sqrt{2}\pi\xi)\,[q_{1}(\xi,0)]^{Wick}\,\Psi^{(N)}\rangle-\int_{\mathcal
  Z_{0}}e^{2i\pi \Re
  \langle \xi,z\rangle}q_{1}(\xi,0)[z]\;d\mu(z)\right|&\leq&
  {\mathcal{A}^{(m)}+\mathcal{B}^{(m)}+\mathcal{C}^{(m)}},
\end{eqnarray*}
with
$$
\mathcal{A}^{(m)}=\left|\langle \Psi^{(N)},\bigg(\mathcal W(\sqrt{2}\pi\xi)-\mathcal W(\sqrt{2}\pi\xi_m)
\bigg) q_{1}(\xi,0)^{Wick}\Psi^{(N)}\rangle\right|,
$$
$$
\mathcal{B}^{(m)}=\left|\langle\Psi^{(N)},\mathcal W(\sqrt{2}\pi\xi_m) \,q_{1}(\xi,0)^{Wick}\Psi^{(N)}\rangle-\mu(e^{2i\pi \Re \langle\xi_m,z\rangle} q_{1}(\xi,0)[z])\right|,
$$
and
$$
\mathcal{C}^{(m)}=\left|\mu(e^{2i\pi\Re \langle \xi_m,z \rangle} q_{1}(\xi,0)[z])-\mu(e^{2i\pi\Re\langle\xi,z\rangle} q_{1}(\xi,0)[z])\right|.
$$
So using Number-Weyl  estimates in \cite[Lemma 3.1]{AmNi1}, one shows that $\mathcal{A}^{(m)}
\lesssim ||\xi-\xi_m||_{\mathcal Z_0}$  and hence $\mathcal{A}^{(m)}\to 0$. Now,
 $\mathcal{B}^{(m)}\to 0$ by the result proved above and $\mathcal{C}^{(m)}\to 0$  by
\eqref{bound1}-\eqref{bound2} and the dominated convergence theorem.
\end{proof}

\subsection{Existence of Wigner measures for all times}
Wigner measures and their properties  were studied in infinite dimensional spaces in \cite{AmNi1}. A result proved in \cite[Theorem 6.2]{AmNi1} says that for any sequence of normal states $\{\widetilde{\varrho}_N(t)\}_{N\in\N}$ as in \eqref{varrhotilde} we can extract a subsequence $(N_k)_{k\in\N}$  such that $\widetilde{\varrho}_{N_{k}}(t)$
has a unique Wigner measure $\tilde\mu_t$ according to Definition \ref{de.wigmeas}. However, the subsequence may depend in the time $t\in\R$. So, in order to carry on the limit on the integral equation \eqref{eq.derivation} we need to extract a subsequence $(N_k)_{k\in\N}$  for all
$t\in\R$ that gives $\mathcal{M}(\widetilde{\varrho}_{N_{k}}(t), k\in \N)=\{\tilde\mu_t\}$.

\begin{prop}
\label{pr.integral}
Let $\{\varrho_{N}=|\Psi^{(N)}\rangle \langle \Psi^{(N)}|\}_{N \in \N}$ be a sequence of normal states on $\vee^{N}\mathcal Z_{0}$ such that
\begin{equation*}
\exists  C>0, \forall N\in\N, \;\langle \Psi^{(N)},H_{N}^{0}\Psi^{(N)}\rangle \leq{CN}\,,
\end{equation*}
and $\mathcal M(\varrho_{N}, N \in \N)=\{ \mu_{0} \}$.
Then for any $\xi \in Q(A)$ and for any subsequence $(N_{k})_{k \in \N}$ there exist a family of probability measures $(\mu_{t})_{t \in \R}$ on $\mathcal{Z}_{0}$ and a subsequence $(N_{k_{l}})_{l \in \N}$ such that for all $t \in \R$,
$$\mathcal{M}\bigg(\big|e^{-itH_{N_{k_{l}}}^{0}}e^{itH_{N_{k_{l}}}}\Psi^{(N_{k_{l}})}\big\rangle \big\langle e^{-itH_{N_{k_{l}}}^{0}}e^{itH_{N_{k_{l}}}}\Psi^{(N_{k_{l}})}\big|,l \in \N\bigg)=\{ \tilde{\mu}_{t}\},$$
and the following  Liouville equation is satisfied for any $\xi \in Q(A)$,
\begin{equation}
\label{eq.liouville}
\begin{aligned}
\tilde{\mu}_{t}(e^{2i\pi \Re \langle \xi,z \rangle})=&\tilde{\mu}_{0}(e^{2i\pi \Re \langle \xi,z \rangle})+i\int_{0}^{t}\tilde{\mu}_{s}(e^{2i\pi\Re \langle\xi,z\rangle}q_{1}(\xi,s)[z])\,ds\\
\displaystyle
=&\tilde{\mu}_{0}(e^{2i\pi \Re  \langle \xi,z
  \rangle})+i\int_{0}^{t}\tilde{\mu}_{s}\big(\big\{q_s(z);e^{2i\pi \Re \langle \xi,z \rangle}\big\}\big)ds,
\end{aligned}
\end{equation}
with $z_s=e^{-isA} z$, $\xi_s=e^{-isA} \xi$, $q_1(\xi,s)=-\pi\, \Im \,q(z_s^{\otimes 2},\mathcal{S}_2\,\xi_s \otimes z_s)$, $q_s(z)=\frac{1}{2}q(z_{s}^{\otimes
  2},z_{s}^{\otimes 2})$  and the bracket  $\{b_{1};b_{2}\}(z)$ equals to $\partial_{z}b_{1}(z)\cdot\partial_{\bar
z}b_{2}(z)-\partial_{z}b_{2}(z)\cdot\partial_{\bar z}b_{1}(z)$.
\end{prop}
\begin{proof}
The extraction of such subsequence $(N_{k_{l}})_{l\in\N}$ for all times follows by an Ascoli type  argument proved in \cite[Proposition 3.3]{AmNi3}. Here we briefly check the main points.  Wigner measures are identified through \eqref{eq.wignercv}. Hence we consider the quantities:
$$
G_{N}(t,\xi)=\langle \widetilde{\Psi}_{t}^{(N)},\mathcal W(\sqrt{2}\pi \xi)\,\widetilde{\Psi}_{t}^{(N)}\rangle.
$$
We wish to prove the existence of a subsequence $(N_{k_{l}})_{l\in\N}$
such that  $G_{N_{k_{l}}}(t,\xi)$ converges
for all $t\in\R$ and $\xi\in \mathcal Z_0$. For this, we exploit the regularity of the functions
$G_{N}(t,\xi)$ with respect to $t$ and $\xi$. In some sense we have to prove that the family
$(G_N)_{N\in\N}$ is equi-continuous on bounded sets of $\R\times\mathcal Z_0$.
By using Lemma 3.1 in \cite{AmNi1} we get for $\xi,\eta \in Q(A)$,
\begin{equation*}
\|\big[\mathcal W(\sqrt{2}\pi \xi)-\mathcal W(\sqrt{2}\pi
\eta)\big](\mathbf{N}+1)^{-\frac{1}{2}}\|_{\mathcal{L}(\Gamma_{s}(\mathcal
  Z_{0}))}\lesssim\|\xi-\eta\|_{\mathcal Z_{0}}\;\sqrt{\|\xi\|^{2}_{\mathcal Z_{0}}+\|\eta\|^{2}_{\mathcal Z_{0}}+1}\,.
\end{equation*}
Therefore, the following estimate holds
\begin{equation}
\label{eq.gn1}
|G_{N}(t,\xi)-G_{N}(t,\eta)|\lesssim\|\xi-\eta\|_{\mathcal Z_{0}}\;\sqrt{\|\xi\|^{2}_{\mathcal Z_{0}}+\|\eta\|^{2}_{\mathcal Z_{0}}+1}\,.
\end{equation}
On the other hand by using Proposition \ref{pr.ipp}, Proposition \ref{c-prop-wickA} (iii) and
Proposition \ref{pr.invsew}, we get for any $s,t \in \R,\xi \in Q(A)$ and $\varepsilon=\frac{1}{N}$,
\begin{align*}
|G_{N}(s,\xi)-G_{N}(t,\xi)|\leq &\bigg|\int_{s}^{t}\langle \widetilde{\Psi}_{r}^{(N)},\mathcal W(\sqrt{2}\pi \xi)\sum_{j=1}^{4} \varepsilon^{j-1} \,q_{j}(\xi,r)^{Wick} \;\widetilde{\Psi}^{(N)}_{r}\rangle \, dr\bigg|&\\
\lesssim & (1+\|\xi\|_{Q(A)}^{4}) \; |s-t|\; \sup_{s\leq r\leq t} \|(A_1+1)^{\frac{1}{2}}\widetilde{\Psi}_{r}^{(N)}
\|^{2}_{\bigvee^{N}\mathcal Z_{0}} \;\lesssim \;(1+\|\xi\|_{Q(A)}^{4})\, |s-t|\,.&\\
\end{align*}
Hence combining \eqref{eq.gn1} with the latter inequality one gets for any  $\eta,\xi \in Q(A)$ and $s,t \in \R$,
\begin{align*}
|G_{N}(t,\xi)-G_{N}(s,\eta)|\lesssim |s-t|(1+\|\xi\|^{4}_{Q(A)})+\|\xi-\eta\|_{\mathcal Z_{0}}\;\sqrt{\|\xi\|^{2}_{\mathcal Z_{0}}+\|\eta\|^{2}_{\mathcal Z_{0}}+1}\,.
\end{align*}
Furthermore the uniform estimate $|G_{N}(t,\xi)|\leq 1$ holds true. By an Ascoli type argument as in \cite[Proposition 3.3]{AmNi3} and \cite[Proposition 3.9]{AmNi4}, we see that  for any  sequence $(N_{k})_{k \in \N}$, there exists a subsequence $(N_{k_{l}})_{l \in \N}$ and a family of Borel probability measures $(\tilde{\mu}_{t})_{t \in \R}$ on $\mathcal Z_{0}$ satisfying for any $t \in \R$,
\begin{equation*}
\mathcal{M}\left(\big|\widetilde{\Psi}_{t}^{(N_{k_{l}})}\big\rangle \big\langle \widetilde{\Psi}_{t}^{(N_{k_{l}})}\big|,\,l \in \N\right)=\{ \tilde{\mu}_{t}\}\,.
\end{equation*}
Now to prove the integral equation \eqref{eq.liouville}, we use Proposition \ref{pr.ipp} with $\varepsilon=\frac{1}{N_{k_{l}}}$,
\begin{equation}
\mathcal J_{N_{k_{l}}}(t)=\mathcal J_{N_{k_{l}}}(0)+i\int_{0}^{t}\langle \widetilde{\Psi}_{s}^{(N_{k_{l}})},\mathcal W(\sqrt{2}\pi\xi)\big[\sum_{j=1}^{4}(\varepsilon^{j-1}
q_{j}(\xi,s)^{Wick}\big]\,\widetilde{\Psi}_{s}^{(N_{k_{l}})}\rangle \; ds,
\end{equation}
with the monomials $(q_{j}(\xi,s))_{j=1,2,3,4}$ given by \eqref{mon-qjs}.
The  estimates provided by Proposition \ref{c-prop-wickA} (iii) and Proposition \ref{pr.invsew} give the convergence towards $0$ of the terms involving $q_{j}(\xi,s)^{Wick}$, $j=2,3,4$ when $l\to\infty$. Applying the Proposition \ref{pr.compactness} to the subsequence $\big|\widetilde{\Psi}_{s}^{(N_{k_{l}})}\big\rangle \big\langle \widetilde{\Psi}_{s}^{(N_{k_{l}})}\big|$, we obtain the claimed equation \eqref{eq.liouville}. Remark that in order to check the hypothesis \eqref{assN} of  Proposition \ref{pr.compactness} we have used Proposition \ref{pr.invhn}.
\end{proof}

\section{The Liouville equation}
\label{se.liouville}
Once Proposition \ref{pr.compactness} is proved we are led to the problem of solving
a Liouville (or continuity) equation in infinite dimension which already admits measure-valued solutions. So the point is to prove uniqueness. The method we use for uniqueness here is introduced in \cite{AmNi4}, and uses some techniques from optimal transport theory initiated in the book \cite{AGS}. Here we rely in the recent improvement in \cite{W1}, briefly recalled in Appendix B.

\subsection{Properties of measure-valued solutions to Liouville equation}
 We need some preliminaries. The sets of all Borel probability measures on $Q'(A)$ and $Q(A)$ are denoted by  $\mathfrak{P}(Q'(A))$ and $\mathfrak{P}(Q(A))$ respectively.   We introduce some classes of cylindrical functions on $Q'(A)$. Denote $\P$ the space of finite rank orthogonal  projections on $Q'(A)$.
 We say that a function $f$ is in the cylindrical Schwartz space $\mathcal S_{cyl}(Q'(A))$ (resp.
$\mathcal C_{0,cyl}^{\infty}(Q'(A))$) if:
 \begin{equation*}
\exists \mathfrak{p} \in \P, \; \exists g \in  \mathcal S(\mathfrak{p} Q'(A))\;\;\big(\textit{resp. $\mathcal C_{0,cyl}^{\infty}(\mathfrak{p} Q'(A))$}\big),\; \forall z \in Q'(A), f(z)=g(\mathfrak{p} z).
\end{equation*}
The space $\mathcal C_{0,cyl}^{\infty}(\R \times Q'(A))$ of smooth cylindrical
functions with compact support on $\R \times Q'(A)$ will be useful too and it is defined in the same way.
Denote $L_{\mathfrak p}(dz)$ the Lebesgue measure on the finite
dimensional subspace $\mathfrak p Q'(A)$. The Fourier transform of functions in $\mathcal S_{cyl}(Q'(A))$ are given by
\begin{equation*}
\mathcal F[f](\xi)=\int_{\mathfrak p Q'(A)} f(z) e^{-2i\pi \Re
  \langle z,\xi \rangle_{Q'(A)}}L_{\mathfrak p}(dz),
\end{equation*}
After fixing a Hilbert basis $(e_{n})_{n \in \N^{*}}$, the space
$Q'(A)$ as a real Hilbert space, can be equipped with a  useful norm,
$$
 \|z\|_{Q'(A),w}=\sqrt{\sum_{n \in \N^{*}}\frac{|\Re \langle z,e_{n}\rangle_{Q'(A)}|^{2}}{n^{2}}}.
$$
The norms $\|.\|_{Q'(A)}$ and $\|.\|_{Q'(A),w}$ lead to  two distinct notions
of narrow convergence  of probability measures. On the one hand, a sequence $(\mu_{n})_{n \in \N}$ is narrowly convergent to $\mu \in \mathfrak{P}(Q'(A))$ if
\begin{equation}
\label{narrow1}
\lim_{n \to
  +\infty}\int_{Q'(A)}f(z)d\mu_{n}(z)=\int_{Q'(A)}f(z)d\mu(z),
\end{equation}
for every function $f \in
\mathcal{C}^{0}_{b}(Q'(A),\|.\|_{Q'(A)})$, the space of continuous and bounded
real functions defined on $(Q'(A),\,\|\cdot\|_{Q'(A)})$. On the other hand, a sequence $(\mu_{n})_{n \in \N}$ is
weakly narrowly convergent if  the limit \eqref{narrow1} holds for
all $f \in \mathcal{C}^{0}_{b}(Q'(A),\|.\|_{Q'(A),w})$. The family of probability measures $(\tilde\mu_t)_{t\in\R}$ provided by Proposition \ref{pr.integral}  have uniformly bounded moments $\int_{Q'(A)}\|z\|_{\mathcal Z_{0}}^{2k}d\tilde\mu_{t}(z)\leq 1$ for all $k\in\N$ thanks to Proposition \ref{pr.transportapriori}. We refer to \cite[Chapter V]{AGS} or \cite{W1} for a more complete presentation of those notions.
\begin{prop}
\label{pr.liouville}
Let  $\{|\Psi^{(N)}\rangle \langle \Psi^{(N)}|\}_{N \in \N}$ a sequence of normal states in $\bigvee^{N}\mathcal Z_0$ satisfying the uniform estimate:
$$
\exists C>0\,,\forall N\in\N, \;  \langle\Psi^{(N)},H_{N}^{0}\Psi^{(N)}\rangle \leq CN\,.
$$
Consider an extracted subsequence $(N_{k})_{k \in \N}$ according to Proposition
\ref{pr.integral} such that  for any $t \in \R$,
$$
\mathcal{M}(|\widetilde{\Psi}_{t}^{(N_{k})}\rangle \langle \widetilde{\Psi}_{t}^{(N_{k})}|,k \in \N)=\{ \tilde{\mu}_{t} \}\,,
$$
where $\widetilde{\Psi}_{t}^{(N_{k})}$ is given by \eqref{varrhotilde}. Then the Borel probability measures $\tilde{\mu}_{t}$ on $\mathcal Z_0$ satisfy:
\begin{itemize}
\item [(i)] $\tilde{\mu}_{t}$ are Borel probability measures on $Q(A)$ carried on $Q(A)$, i.e. $\tilde{\mu}_{t}(Q(A))=1$.
\item [(ii)] The map $t \mapsto \tilde{\mu}_{t}$ is weakly narrowly continuous in $\mathfrak{P}(Q'(A)).$
\item [(iii)] The measure $\tilde{\mu}_{t}$ is a weak solution of the Liouville
equation
$$
\partial_{t}\tilde{\mu}_{t}+i\{ q_t(z); \tilde{\mu}_{t} \}=0,
$$
\end{itemize}
i.e: for all $f \in \mathcal C^{\infty}_{0,cyl}(\R \times Q'(A))$
\begin{equation}
\label{liouv1}
\int_{\R}\int_{Q(A)}\partial_{t}f(t,z)+i\{q_t(.),f(t,.)\}(z)\;d\tilde{\mu}_{t}(z)\;dt=0\,,
\end{equation}
where  $z_{t}=e^{-itA}z$ and  $q_t(z)=\frac{1}{2}q(z_{t}^{\otimes
  2},z_{t}^{\otimes 2})$.
\end{prop}

\begin{proof}
The statement (i) is proved in \cite[Proposition 3.11]{AmNi4} when $A=-\Delta$ but the proof works without any change for a general operator $A$ satisfying \eqref{A1}. The proof of  the statements (ii)-(iii) are also essentially the same as in \cite[Proposition 3.14]{AmNi4}. We briefly sketch here the main arguments.\\
(ii) {\it{Weakly narrowly continuity:}} \\
The characteristic function of $\tilde\mu_t$ as a probability measure on $Q'(A)$ is given by
$$
G(t,\xi)=\tilde\mu_{t}(e^{2i\pi \Re \langle \xi,z\rangle_{Q'(A)}}):=\tilde\mu_{t}(e^{2i\pi \Re \langle \xi,(A+1)^{-1}z\rangle_{\mathcal Z_{0}}}).
 $$
The following inequality holds as in \cite[Proposition 3.11]{AmNi4} for any $\xi,\xi' \in Q'(A)$,
\begin{equation}
|G(t,\xi)-G(t,\xi')|\leq \pi \|\xi-\xi'\|_{Q'(A)}\int_{\mathcal Z_{0}} \|z\|^{2}_{Q(A)}\,d\tilde{\mu}_{t}(z).
\end{equation}
Since by Lemma \ref{pr.invhn} there exists a time independent  constant  $C'>0$ such that
$\langle \Psi_{t}^{(N)},H_{N}^{0}\Psi_{t}^{(N)}\rangle\leq C'N$, one obtains using
Proposition \ref{pr.transportapriori} the uniform estimate,
\begin{equation}
\label{tight}
\int_{\mathcal
  Z_{0}}\|z\|^{2}_{Q(A)}\,d\tilde{\mu}_{t}(z)\leq  C'\,.
\end{equation}
Subsequently for any $\xi,\xi' \in Q'(A)$,
\begin{equation}
\label{eq.glipschitz}
|G(t,\xi)-G(t,\xi')|\lesssim \|\xi-\xi'\|_{Q'(A)}.
\end{equation}
On the other hand for any $\xi \in Q'(A)$ and $t,t' \in \R$, the following estimate holds true
\begin{equation}
\label{eq.char}
|G(t',\xi)-G(t,\xi)|\leq \big|\int_{t'}^{t}\tilde{\mu}_{s}(e^{2i\pi\Re \langle\xi,(A+1)^{-1} \,z\rangle}q_{1}(\xi,s)[z])\,ds\big| \leq
(C'+1)\;|t-t'|\;\|\xi\|_{Q'(A)},
\end{equation}
owing to \eqref{C1} and Proposition \ref{pr.transportapriori}.
Now let $g\in \mathcal S_{cyl}(Q'(A))$ based on $\mathfrak p Q'(A)$ and
\begin{equation*}
I_{g}(t):=\int_{\mathfrak p Q'(A)} g(z)
\;d\tilde{\mu}_{t}(z)=\int_{\mathfrak p Q'(A)}\mathcal
F[g](\xi)\;G(\xi,t)\, L_{\mathfrak p}(d\xi).
\end{equation*}
Then we easily check:
\begin{itemize}
\item
$t \longrightarrow \mathcal F[g](\xi)\;G(t,\xi)$ is continuous
owing to \eqref{eq.char}.
\item
$\xi \longrightarrow \mathcal F[g](\xi)\;G(t,\xi)$ is bounded
by a $L_{\mathfrak p}(d\xi)$-integrable function.
\end{itemize}
\bigskip
Thus  $I_{g}(\cdot)$  is continuous for all $g \in \mathcal S_{cyl}(Q'(A))$ and the bound \eqref{tight} holds true. Hence we can apply Lemma 5.12-f) in \cite{AGS} and then conclude that the map $t \rightarrow \tilde{\mu}_{t}$ is weakly narrowly continuous in $Q'(A)$.

\medskip
\noindent
{\it{The Liouville equation:}}\\
Integrate the expression \eqref{eq.liouville} with $\mathcal
F[g](\xi)L_{\wp}(dz)$, hence  $\forall t \in \R$, $\forall g \in \mathcal S_{cyl}(Q'(A))$,
\begin{equation*}
\partial_{t}I_{g}(t)=i\int_{Q(A)}\{q_{t};g\}(z) d\tilde{\mu}_{t}(z),
\end{equation*}
with $q_{t}(z)=\frac{1}{2}q(z_{t}^{\otimes 2},z_{t}^{\otimes 2}).$
Multiplying this expression by $\phi \in \mathcal C^{\infty}_{0}(\R)$ and
integrating  by parts yields
\begin{equation*}
 \int_{\R}\int_{Q(A)}\partial_{t}f(t,z)+i \{ q_{t}(.),f(t,.)\}(z) \;d\tilde\mu_{t}(z)dt=0\,,
\end{equation*}
with $f(t,z)=g(z)\phi(t)$.
To conclude, we use the density of $\mathcal C^{\infty}_{0}(\R) \otimes^{alg}
\mathcal C^{\infty}_{0,cyl}(Q'(A))$ in $\mathcal C^{\infty}_{0,cyl}(\R \times Q'(A)).$
\end{proof}
\begin{remarks}
We will show in the following section that the Liouville equation \eqref{eq.liouville} is equivalent to the following one for every $f \in \mathcal{C}_{0,cyl}^{\infty}(\R \times Q'(A))$
\begin{equation}
\label{liouv2}
 \int_{\R}\int_{Q(A)}\partial_{t}f(t,z)+ \Re \,\langle v_t(z),\nabla f(t,z)\rangle_{Q'(A)} \;d\tilde\mu_{t}(z)dt=0\,,
\end{equation}
with $v_t(z)=-ie^{itA}[\partial_{\bar z}q_{0}](e^{-itA}z),$ and $\nabla$ is the real derivative in $Q'(A).$ 
\end{remarks}

\subsection{End of the Proof of Theorems \ref{thm.main}}
\label{se.ep}
\begin{proof}
\label{pr.end}
Assume the hypotheses of Theorem \ref{thm.main} and consider  for a given time  $t \in \R$ the family of normal states,
$$\tilde{\varrho}_{N}(t)=|\widetilde{\Psi}_{t}^{(N)}\rangle
\langle
\widetilde{\Psi}_{t}^{(N)}|=|e^{itH_{N}^{0}}e^{-itH_{N}}\Psi^{(N)}\rangle
\langle e^{itH_{N}^{0}}e^{-itH_{N}}\Psi^{(N)}|\,.
$$
Suppose that $\nu$ is any Wigner measure of $\tilde{\varrho}_{N}(t)$ then there exists a subsequence $(N_k)_{k\in\N}$ such that $\{\nu\}=\mathcal{M}(\varrho_{N_k}(t),k\in\N)$ according to Definition \ref{de.wigmeas}. By Proposition \ref{pr.integral} and  \ref{pr.liouville}, we can extract a subsequence $(N_{k_{l}})_{l\in\N}$ such that for all $s\in\R$,
$$
\mathcal{M}(\tilde\varrho_{N_{k_{l}}}(s),l\in\N)=\{\tilde\mu_s\}\,\quad \text{ with in particular  } \quad
\tilde\mu_t=\nu\,.
$$
We know by Proposition \ref{pr.liouville} that $s\in\R\to\tilde\mu_s$ solves the Liouville equation \eqref{eq.liouville} and by setting $\xi=(1+A)^{-1}\eta,\,\eta \in Q'(A),$ we get
\begin{equation}
\tilde{\mu}_{t}(e^{2i\pi \Re \langle \eta,z \rangle_{Q'(A)}})=\tilde{\mu}_{0}(e^{2i\pi \Re \langle \eta,z \rangle_{Q'(A)}})+i\int_{0}^{t}\tilde{\mu}_{s}(e^{2i\pi\Re \langle \eta,z\rangle_{Q'(A)}}q_{1}(\eta,s)[z])\,ds,\\
\end{equation}
 Then, we have by integrating on $\R \times Q'(A)$
\begin{equation*}
 \int_{\R}\int_{Q(A)}\partial_{t}f(t,z)+ \Re \, \langle v_t(z),\nabla f(t,z) \rangle_{Q'(A)} \;d\tilde\mu_{t}(z)\,dt=0\,,
\end{equation*}
for any $f \in \mathcal C^{\infty}_{0,cyl}(\R\times Q'(A))$, owing to the following equalities
\begin{align*}
i \{ q_{s}(.),f(s,.)\}(z)&=i \langle [\partial_{\bar z}q_0](e^{-isA}z),e^{-isA}\partial_{\bar z}f(s,z)\rangle-i\langle e^{-isA}\partial_{\bar z}f(s,z),[\partial_{\bar z}q_0](e^{-isA}z)\rangle\\
                         &=2\,\Re \langle v_s(z), \partial_{\bar z}f(s,z)\rangle=2\, \Re \langle v_s(z), \nabla_{\bar z}f(s,z)\rangle_{Q'(A)}\\
&=\Re \langle v_{s}(z),\nabla f(s,z)\rangle_{Q'(A)},
\end{align*}
where $v_s(z)=-ie^{isA}[\partial_{\bar z}q_0](e^{-isA}z)$, $q_{s}(z)=\frac{1}{2} q(z_{s}^{\otimes 2},z_{s}^{\otimes 2}),\;z_{s}=e^{-isA}z$. Here $v_s$ has the interpretation of a velocity vector field and $\nabla $ is the real derivative in $Q'(A)$, see \cite[Lemma C.7]{AmNi4} for more details. By Proposition \ref{pr.invhn}, we see that for any $s \in \R$,
\begin{equation*}
\langle \widetilde{\Psi}_{s}^{(N)},H_{N}^{0}\widetilde{\Psi}_{s}^{(N)}\rangle \leq C'N\,,
\end{equation*}
for some time independent constant $C'>0$. Thus, Proposition \ref{pr.transportapriori} gives
for any $s \in \R$, and for every $k \in \N$
\begin{equation*}
\int_{Q(A)} \|z\|_{Q(A)}^{2}\;\|z\|_{\mathcal{Z}_{0}}^{2k}\; d\tilde{\mu}_{s}(z) \leq C'\,.
\end{equation*}
Subsequently, for any time $t \in \R, \tilde\mu_t(B_{\mathcal Z_0}(0,1))=1.$ Now the abstract mean-field  equation
\begin{equation*}
i\partial_{t} z =Az+[\partial_{\bar z}q_0](z),
\end{equation*}
can be written in the interaction representation as follows:
\begin{equation}
\label{eq.hartree2}
\left\{
      \begin{aligned}
        &\partial_{t}z=v_{t}(z)= -ie^{itA}[\partial_{\bar{z}}q_0](e^{-itA}z),&\\
        &z_{|t=0}=z_{0}.\\
         \end{aligned}
    \right.
\end{equation}
So the above equation \eqref{eq.hartree2} is locally well-posed in $Q(A)$ thanks to the assumption \eqref{C1}.
Remember that Proposition \ref{pr.liouville} says that the map $s \to
\tilde{\mu}_{s}$ is weakly narrowly continuous in $Q'(A).$
Subsequently, the measures $(\tilde{\mu}_{s})_{s\in\mathbb{R}}$ are satisfying all the assumptions of Theorem \ref{pr.D5}. Then we get
$$
\forall s \in I,\,\tilde{\mu}_{s}=\tilde{\Phi}(s,0)_\sharp\mu_{0},
$$
where $\tilde{\Phi}(s,0)$ denotes the well defined flow of the equation \eqref{eq.hartree2} and $I$ is the interval provided by \eqref{C1}.
In particular one gets the equality $\nu=\tilde{\Phi}(t,0)_\sharp\mu_{0}$. Since $\nu$ is any Wigner measure of $(\tilde \varrho_N(t))_{N\in\N}$, one obtains
$$
\mathcal{M}(\tilde\varrho_N(t),N\in\N)=\{\tilde{\Phi}(t,0)_\sharp\mu_{0}\}\,.
$$
Back to the family of normal states,
$$
\varrho_N(t)=e^{-itH_{N}^{0}}\,\tilde\varrho_N(t) \,e^{itH_{N}^{0}}\,.
$$
We notice that $e^{itH_{N}^{0}} \mathcal{W}(\xi) e^{-itH_{N}^{0}}= \mathcal{W}(e^{itA}\xi)$ hence a simple computation yields for any $t \in I$,
\begin{equation*}
\mathcal{M}(e^{-itH_{N}^{0}}\,\tilde\varrho_N(t) \,e^{itH_{N}^{0}}
,\, N \in \N) =\{(e^{-itA})_\sharp\nu, \nu\in\mathcal{M}(\tilde\varrho_N(t),\, N \in \N)\}=
\{(e^{-itA})_\sharp(\tilde\Phi(t,0)_{\sharp}\mu_{0})\}.
\end{equation*}
Finally, remark that $\Phi(t,0)=e^{-itA}\circ\tilde\Phi(t,0)$. So, the main Theorem \ref{thm.main} is now proved.
\end{proof}

\appendix
\section{The Wick quantization and Wigner measures}
\label{se.app}
Although the mean-field problem considered in this paper deals with many-body Schr\"odinger Hamiltonians of the form of $H_N$ given in \eqref{eq.hn}, it is conceptually important to see $H_N$ as a second quantization of the classical energy \eqref{cl-enr}. This  in particular allows to understand the phase-space analysis hidden in the mean-field approximation and provides convenient tools to analyze  phase-space distributions of states as well as their evolutions.

\subsection{Wick quantization}
So, for reader's convenience  we briefly recall in this appendix the $\varepsilon$-dependent Wick quantization in the Fock spaces and provide some general  properties. Let $\mathcal{Z}_0$ be a complex Hilbert space and consider the symmetric Fock space
\begin{equation*}
\Gamma_{s}(\mathcal{Z}_{0})=\bigoplus_{n=0}^{\infty}\mathcal{S}_{n}(\mathcal{Z}_{0}^{\otimes n})=\bigoplus_{n=0}^{\infty}\vee^{n}
\mathcal{Z}_{0}\,,
\end{equation*}
where $\mathcal{S}_{n}$ denotes the symmetrization operator given by \eqref{sym}. On this Fock space there exists a realization of the following $\varepsilon$-dependant canonical commutation relations (CCR):
\begin{equation*}
\big[a(z_{1}), a^{*}(z_{2})\big]=\varepsilon \langle z_{1}\,,\,
z_{2}\rangle \;\Id,\qquad  \big[a^{*}(z_{1}), a^{*}(z_{2})\big]=\big[a(z_{1}),a(z_{2})\big]=0\,, \;\;\varepsilon>0\,,
\end{equation*}
given by the $\varepsilon$-dependent annihilation and creation operators,
\begin{eqnarray*}
a(z_{1})_{|\vee^{N}\mathcal Z_{0}}&=&\sqrt{\varepsilon N}\langle z_{1}| \otimes
\Id_{|\vee^{N-1}\mathcal Z_{0}},\\
a^{*}(z_{2})_{|\vee^{N}\mathcal Z_{0}}&=&\sqrt{\varepsilon (N+1)}\,\mathcal{S}_{N+1}(|z_{2}\rangle \otimes
\Id_{|\vee^{N}\mathcal Z_{0}})\,.
\end{eqnarray*}
The  Weyl operator is also $\varepsilon$-dependent and it is defined for any $\xi\in \mathcal Z_{0}$ by the formula:
\begin{equation}
\label{weyl}
\mathcal{W}(\xi)=e^{i\frac{a(\xi)+a^{*}(\xi)}{\sqrt{2}}}\,.
\end{equation}
For $i=1,\cdots,n$ and $C$ an operator on $\mathcal{Z}_0$, we denote
$$
C_i=1^{\otimes (i-1)} \otimes
C\otimes 1^{\otimes (n-i)},
$$
where the operator $C$ in the right hand side acts on the $i^{th}$ component.
The second quantization  $d\Gamma(C)$ is the $\varepsilon$-dependent operator  defined by
\begin{equation*}
d\Gamma(C)_{|\vee^{n}\mathcal Z_{0}}=\varepsilon\sum_{i=0}^{n} C_i\,.
\end{equation*}
In particular, the $\varepsilon$-dependent number operator is
\begin{equation}
\label{dfn.number}
\mathbf{N}:=d\Gamma(\Id)\,.
\end{equation}
The Wick quantization is a map that corresponds to a monomial $z\in\mathcal{Z}_0
\mapsto b(z)$  an operator on the Fock space (the function $b(z)$ is called a symbol in connection with the pseudo-differential calculus). It is related to the normal ordering of  products of creation-annihilation operators which is a well treated subject in standard textbooks (see for instance \cite{Ber,DeGeBook}). Here we follow the  presentation in \cite{AmNi1} which stresses the symbol-operator correspondence and which is more convenient for our purpose. We  introduce below two type of classes of symbols $\mathcal P_{p,q}(\mathcal Z_0)$ and $\mathcal Q_{p,q}(A)$ and stress their main properties.

For all $p,q \in \N$, we denote $\mathcal P_{p,q}(\mathcal Z_0),$ resp. $\mathcal P^{\infty}_{p,q}(\mathcal Z_0),$ the space of complex-valued monomials on $\mathcal Z_{0}$, defined according to the conditions:
\begin{equation}
\label{def.symb}
\begin{aligned}
b \in \mathcal P_{p,q}(\mathcal Z_0)& \Leftrightarrow \exists! \,\tilde{b} \in
\mathcal{L}(\vee^{p}\mathcal Z_{0},\vee^{q}\mathcal Z_{0}),\;
b(z)=\langle z^{\otimes q},\tilde{b} z^{\otimes p}\rangle\,,\\
b \in \mathcal P_{p,q}^{\infty}(\mathcal Z_0)& \Leftrightarrow \exists! \,\tilde{b} \in
\mathcal{L}^{\infty}(\vee^{p}\mathcal Z_{0},\vee^{q}\mathcal Z_{0}),\;
b(z)=\langle z^{\otimes q},\tilde{b} z^{\otimes p}\rangle.
\end{aligned}
\end{equation}
Here $\mathcal{L}$ and $\mathcal{L}^{\infty}$  refer to the space of
bounded operators and the space of compact operators respectively.
\begin{dfn}
\label{dfn.wick} For $\varepsilon>0$ and for each symbol $b\in \mathcal P_{p,q}(\mathcal Z_0)$, with $\tilde b$ as in \eqref{def.symb},  we associate an operator $b^{Wick}$:
$\bigoplus_{n\geq 0}^{alg}\vee^{n}\mathcal
Z_{0} \longrightarrow \bigoplus_{n\geq 0}^{alg}\vee^{n}\mathcal
Z_{0}$, given by
\begin{equation}
\label{def.wicko}
b_{|\vee^{n}\mathcal Z_{0}}^{Wick}=1_{[p,+\infty)}(n)
\frac{\sqrt{n!(n+q-p)!}}{(n-p)!}\varepsilon^{\frac{p+q}{2}}\mathcal{S}_{n-p+q}(\tilde{b}\otimes
\Id^{\otimes (n-p)})\in \mathcal{L}(\vee^{n}\mathcal
Z_{0},\vee^{n+q-p}\mathcal Z_{0})\;.
\end{equation}
\end{dfn}
The  Wick quantization map depends in the parameter $\varepsilon>0$, however for simplicity we omit this dependence in the notation of $b^{Wick}$. By linearity one can extend this quantization to any finite sum in $\mathcal{P}_{alg}(\mathcal Z_0):=\oplus_{p,q\geq 0}^{alg}\mathcal{P}_{p,q}(\mathcal Z_0)$. Remark however that the classical energy functional $h(z)=\langle z,Az \rangle+ \frac{1}{2}q(z^{\otimes 2},z^{\otimes 2})$ (given in \eqref{cl-enr}) is not in $\mathcal{P}_{alg}(\mathcal Z_0)$ unless $A$ and $q$ are bounded. So, in order to extend the above quantization procedure to more interesting symbols, we introduce in the sequel another class of monomials $\mathcal{Q}_{p,q}(A)$.

Let $A$ be a given non-negative self-adjoint operator on $\mathcal Z_0$. Let $H_n^0$ denotes, for  each $n\in\N$, the operator on $\vee^{n}\mathcal{Z}_0$ defined according to \eqref{ham0}, i.e.:
$$
H^0_{n_{|\vee^{n}\mathcal{Z}_0}}=\sum_{i=1}^n A_i\,.
$$
For simplicity we denote
$$
\mathfrak{Q}_n:=Q(H_n^0)\subset \vee^{n}\mathcal{Z}_0 \quad \text{ and }
\quad Q_n:=Q(\sum_{i=1}^n A_i)\subset \otimes^{n}\mathcal{Z}_0\,,
$$
with $Q_n$ is a subspace possessing non symmetric vectors satisfying $\mathfrak{Q}_n\subset Q_n$, $\mathcal{S}_n Q_n=\mathfrak{Q}_n$ and $Q_n, \mathfrak{Q}_n$ are respectively dense in
 $\otimes^n\mathcal Z_0, \vee^n\mathcal Z_0$. Remember that $Q_n$ and $\mathfrak{Q}_n$ are Hilbert spaces when they are equipped with the graph norm
\begin{equation}
\label{inner}
\|u\|_{Q_n}=\|u\|_{\mathfrak{Q}_n}=\sqrt{\langle u, \sum_{i=1}^n [A_i+1]\, u\rangle}\,,\quad \forall
u\in Q_n.
\end{equation}
We denote by $Q_n'$ and $\mathfrak{Q}_n'$ respectively the dual spaces of $Q_n$ and $\mathfrak{Q}_n$ with respect to the scalar product of $\otimes^n\mathcal{Z}_0$.\\
For all $p,q \in \N$, we define the class of symbols $\mathcal Q_{p,q}(A)$ as the space of complex-valued monomials on $Q(A)$ verifying
\begin{equation}
\label{def.symbq}
b \in \mathcal Q_{p,q}(A) \Leftrightarrow \exists! \,\tilde{b} \in
\mathcal{L}(\mathfrak{Q}_{p},\mathfrak{Q}_{q}'),\; \forall z\in Q(A), \;
b(z)=\langle z^{\otimes q},\tilde{b} z^{\otimes p} \rangle_{\otimes^q\mathcal Z_0}\,.
\end{equation}
Let $b\in  \mathcal Q_{p,q}(A)$ and $\tilde b$ as in \eqref{def.symbq}, then the map defined for any $\varphi_1,\cdots,\varphi_n\in Q(A)$ by
\begin{equation}
\label{ap.btilde}
\tilde b\otimes 1^{(n-p)} \; \,\mathcal{S}_p\otimes 1^{(n-p)} \,\varphi_1\otimes \cdots\otimes\varphi_n=
\left(\tilde b\,\mathcal{S}_p \,(\varphi_1\otimes \cdots\varphi_p)\right)\otimes\varphi_{p+1}\cdots\otimes
\varphi_n\,,
\end{equation}
extends by linearity and continuity  to a bounded operator from
$Q_{n}$ into $Q_{n-p+q}'$ since for any $\Phi^{(n)}\in \otimes^{alg,n} Q(A)$
 \begin{eqnarray*}
\|\tilde b\otimes 1^{(n-p)} \; \,\mathcal{S}_p\otimes 1^{(n-p)}\, \Phi^{(n)}\|_{Q_{n-p+q}'} &=&
\|(\sum_{i=1}^{n-p+q}A_i+1)^{-\frac{1}{2}} \mathcal{S}_q\tilde b \mathcal{S}_p (\sum_{i=1}^{p}A_i+1)^{-\frac{1}{2}} (\sum_{i=1}^{p}A_i+1)^{\frac{1}{2}} \Phi^{(n)}\|\\
&\leq& \|\tilde b\|_{\mathcal{L}(\mathfrak{Q}_{p},\mathfrak{Q}_{q}')} \| \Phi^{(n)}\|_{Q_n}\,,
\end{eqnarray*}
and the subspace $\otimes^{alg,n} Q(A)$ is a form core for $\sum_{i=1}^n A_i$.
As a consequence, we see that
$$
\mathcal{S}_{n-p+q} \,\tilde b\otimes 1^{(n-p)}\,\mathcal{S}_n= \mathcal{S}_{n-p+q} \,\tilde b\otimes 1^{(n-p)} \;\mathcal{S}_p\otimes 1^{(n-p)} \,\mathcal{S}_n\in \mathcal{L}(\mathfrak{Q}_n,\mathfrak{Q}_{n-p+q}')\,.
$$
\begin{dfn}
\label{dfn.wickq}
For each symbol $b\in \mathcal Q_{p,q}(A)$, with $\tilde b$ as in \eqref{def.symbq},  we associate an operator $b^{Wick}$:
$\bigoplus_{n\geq 0}^{alg}\mathfrak{Q}_n \longrightarrow \bigoplus_{n\geq 0}^{alg}\mathfrak{Q}_n'$, given by
\begin{equation}
\label{def.wickoq}
b_{|\mathfrak{Q}_n}^{Wick}=1_{[p,+\infty)}(n)
\frac{\sqrt{n!(n+q-p)!}}{(n-p)!}\varepsilon^{\frac{p+q}{2}}\mathcal{S}_{n-p+q}(\tilde{b}\otimes
1^{\otimes (n-p)})\in \mathcal{L}(\mathfrak{Q}_n,\mathfrak{Q}_{n-p+q}')\;.
\end{equation}
\end{dfn}
Actually $b_{|\mathfrak{Q}_n}^{Wick}$ can also be understood as a bounded sesquilinear form on $\mathfrak{Q}_n\times \mathfrak{Q}_{n-p+q}$. Remark that we have always the inclusion $\mathcal P_{p,q}(\mathcal Z_0)\subset \mathcal Q_{p,q}(A)$. Furthermore, the class $\mathcal Q_{p,q}(A)$ depends on the operator $A$ and if $A$ is bounded on $\mathcal Z_0$ then  $\mathcal Q_{p,q}(A)$ coincides with $\mathcal P_{p,q}(\mathcal Z_0)$.

\bigskip
\noindent
\textit{Examples}:
Let $q$ be a quadratic form on $Q_2$ satisfying the assumption \eqref{A2} and
$\tilde q$ defined according to \eqref{qtilde}. The main examples of interest here are
\begin{align*}
&b_0(z)=\langle z, A z\rangle\in \mathcal Q_{1,1}(A) &\quad& \text{ with } \quad \tilde b_0=A\,,\\
&b(z)=q(z^{\otimes 2},z^{\otimes 2})\in \mathcal Q_{2,2}(A) &\quad& \text{ with } \quad \tilde b=\mathcal S_2 \tilde q\mathcal S_2\,,
\end{align*}
and
\begin{equation}
h(z)=\langle z, A z\rangle+\frac{1}{2}q(z^{\otimes 2},z^{\otimes 2})\in  \mathcal Q_{1,1}(A)+\mathcal Q_{2,2}(A)\,.
\end{equation}
So using the Wick quantization given in Definition \ref{dfn.wickq}, one obtains the following
equality in the sense of quadratic forms for any $\Psi^{(N)},\Phi^{(N)}\in \mathfrak{Q}_N$,
$$
\langle \Psi^{(N)}, H_N\Phi^{(N)}\rangle = \langle \Psi^{(N)},\varepsilon^{-1}
h^{Wick}\Phi^{(N)}\rangle\,, \quad \text{ when } \quad \varepsilon=\frac 1 N\,.
$$
This identity shows the relationship between the  Hamiltonian of many-boson systems in the  mean-field scaling and the Wick quantization of symbols in $\mathcal Q_{p,q}(A)$ with the
semiclassical parameter $\varepsilon$. In fact most of the information we need in the analysis
of the mean-field approximation comes from general properties of the classes  $\mathcal Q_{p,q}(A)$ stated in Proposition \ref{c-prop-wickA} below.

The linear space $\mathcal Q_{p,q}(A)$ is a subset of the space of continuous functions
on $Q(A)$ and can be equipped with a convenient convergence topology. We say that a sequence
$(c_m)_{m\in N}$ in  $\mathcal Q_{p,q}(A)$ is ${\it b}$-convergent to a function $c(z)$ iff
$$
c_m\overset{\it b}{\to} c \Leftrightarrow \forall z\in Q(A), \,c_m(z)\to c(z) \text{ and }
(||\tilde c_m||_{\mathcal{L}(\mathfrak{Q}_p,\mathfrak{Q}_q')})_{m\in\N} \;\text{ is bounded }\,.
$$
\begin{prop}
\label{c-prop-wickA}
For any $b\in \mathcal Q_{p,q}(A)$ and
$(c_m)_{m\in N}$ a sequence in  $\mathcal Q_{p,q}(A)$, we have:\\
(i)  $\bar{b}\in \mathcal Q_{q,p}(A)$ and
$$
\big(b^{Wick}_{|\mathfrak{Q}_n}\big)^*=\bar b^{Wick}_{|\mathfrak{Q}_{n-p+q}}\,.
$$
(ii) For any $t\in\mathbb{R}$, $b_t(z):=b(e^{-i t A } z)\in \mathcal Q_{p,q}(A)$ with
$$
e^{i \frac{t}{\varepsilon} d\Gamma(A) } b^{Wick} e^{-i \frac{t}{\varepsilon} d\Gamma(A) }
=b_t^{Wick}\,.
$$
(iii) There exists a constant $C_{p,q}>0$ such that for any $\Psi^{(n)}\in \mathfrak{Q}_{n}$, $\Phi^{(m)}\in\mathfrak{Q}_{m}$
with $m=n-p+q$ and  $\varepsilon=\frac{1}{n}$,
$$
\left|\langle \Phi^{(m)}, b^{Wick} \,\Psi^{(n)}\rangle\right|\leq
C_{p,q} \, \left\| \tilde b\right\|_{\mathcal{L}(\mathfrak Q_n,\mathfrak Q_m')} \;
\left\|(A_1+1)^{\frac{1}{2}}\Phi^{(m)}\right\| \,
\left\|(A_1+1)^{\frac{1}{2}}\Psi^{(n)}\right\|\,.
$$
(iv) If $c_m\overset{b}{\to} c$ then $c\in \mathcal Q_{p,q}(A)$ and $c_m^{Wick}$ converges weakly to  $c^{Wick}$ in $\mathcal{L}(\mathfrak{Q}_{n},\mathfrak{Q}_{n-p+q}')$.\\
(v) For any $\xi\in Q(A)$ the symbol $b(\cdot +\xi)$ belongs to $\oplus_{p,q\in\mathbb{N}}^{alg} \mathcal Q_{p,q}(A )$ and the identity
\begin{equation}
\label{wicktrans}
b^{Wick} \mathcal W(\frac{\sqrt{2}}{i\varepsilon} \xi )
=\mathcal W(\frac{\sqrt{2}}{i\varepsilon} \xi )  b(z + \xi)^{Wick} \,,
\end{equation}
holds in the sense of sesquilinear forms on $\mathfrak{Q}_{n_1}\times\mathfrak{Q}_{n_2}$ for any
$n_1,n_2\in \N$.
\end{prop}
\begin{proof}
(i)  According to   \eqref{def.symbq}, we have
$$
\bar b(z)=\overline{b(z)}=\langle \tilde{b} \,z^{\otimes q},   z^{\otimes p} \rangle=\langle z^{\otimes q}, \tilde{b}^* \,  z^{\otimes p} \rangle\,,
$$
where $\tilde{b}^*\in\mathcal{L}(\mathfrak{Q}_q, \mathfrak{Q}_p')$ is the adjoint of
$\tilde{b}\in\mathcal{L}(\mathfrak{Q}_p, \mathfrak{Q}_q')$.
Let $\Phi^{(n)}\in\vee^{alg,n}Q(A)$, $\Psi^{(m)}\in\vee^{alg,m}Q(A)$ with $m=n-p+q$, then we have
\begin{eqnarray*}
\langle b^{Wick} \Phi^{(n)}, \Psi^{(m)}\rangle&=&
1_{[p,+\infty)}(n)
\frac{\sqrt{n!m!}}{(n-p)!}\varepsilon^{\frac{p+q}{2}}  \; \langle \tilde{b}\otimes
1^{\otimes (n-p)} \, \Phi^{(n)}, \Psi^{(m)}\rangle\\
&=&
1_{[p,+\infty)}(n)
\frac{\sqrt{n!m!}}{(n-p)!}\varepsilon^{\frac{p+q}{2}}  \; \langle \Phi^{(n)},
\tilde{b}^*\otimes
1^{\otimes (n-p)} \Psi^{(m)}\rangle\\
&=&  \langle  \Phi^{(n)}, \bar b^{Wick}\Psi^{(m)}\rangle\,.
\end{eqnarray*}
Since $\vee^{alg,n}Q(A)$ is dense in the Hilbert space $(\mathfrak{Q}_n,||.||_{\mathfrak{Q}_n})$, the above identity extends to any $\Phi^{(n)}\in \mathfrak{Q}_n$ and  $\Psi^{(m)}\in \mathfrak{Q}_m$.\\
(ii) For any $t\in\mathbb{R}$ and $n\in\mathbb{N}$
 the  operator $(e^{-i t A })^{\otimes n}:\mathfrak{Q}_n \to \mathfrak{Q}_n$ is bounded and
  extends by duality to a bounded operator on $\mathfrak{Q}_n '$. Hence for any $z\in Q(A)$,
$$
b_t(z):=\langle  (e^{-i t A } z)^{\otimes q},  \tilde{b}\, (e^{-i t A } z)^{\otimes p} \rangle
=\langle z^{\otimes q},   (e^{-i t A })^{\otimes q} \;\tilde{b}\, (e^{-i t A })^{\otimes p}  \;z^{\otimes p} \rangle\,,
$$
and $\tilde b_t=(e^{-i t A })^{\otimes q} \;\tilde{b}\, (e^{-i t A })^{\otimes p}$ belongs to
$\mathcal{L}(\mathfrak{Q}_p, \mathfrak{Q}_q')$.  Let $\Phi^{(n)}\in\mathfrak{Q}_n$, $\Psi^{(m)}\in\mathfrak{Q}_m$ with $m=n-p+q$, then we have
\begin{eqnarray*}
\langle \Psi^{(m)},  e^{i \frac{t}{\varepsilon} d\Gamma(A) } b^{Wick} e^{-i \frac{t}{\varepsilon} d\Gamma(A) }\Phi^{(n)} \rangle&=&
1_{[p,+\infty)}(n)
\frac{\sqrt{n!m!}}{(n-p)!}\varepsilon^{\frac{p+q}{2}}  \;   \langle\Psi^{(m) } ,
 (e^{-i t A })^{\otimes m} \; \tilde{b}\otimes
1^{\otimes (n-p)} \, (e^{-i t A })^{\otimes n} \Phi^{(n)}\rangle\\ \displaystyle
&=&
1_{[p,+\infty)}(n)
\frac{\sqrt{n!m!}}{(n-p)!}\varepsilon^{\frac{p+q}{2}}  \; \langle \Psi^{(m)} ,
\tilde{b}_t\otimes
1^{\otimes (n-p)} \Phi^{(n)} \rangle\\
&=&  \langle  \Psi^{(m)} ,  b_t^{Wick}\Phi^{(n)} \rangle\,.
\end{eqnarray*}
(iii) A simple estimate gives
\begin{eqnarray*}
\left|\langle \Phi^{(m)}, b^{Wick} \,\Psi^{(n)}\rangle\right|&\leq&
\frac{\sqrt{n!(n+q-p)!}}{(n-p)!}
\varepsilon^{\frac{p+q}{2}} \,
\left|\bigg\langle (H_q^0+1)^{\frac{1}{2}}\otimes 1^{\otimes(m-q)}\Phi^{(m)}; \right.\\ && \hspace{.1in} \left.\bigg((H_q^0+1)^{-\frac{1}{2}}  \tilde b(H_p^0+1)^{-\frac{1}{2}}\bigg)\otimes 1^{\otimes(n-p)} \,(H_p^0+1)^{\frac{1}{2}}\otimes 1^{\otimes(n-p)}\Psi^{(n)}\bigg\rangle\right|\\
&\leq&
 \, \left\| \tilde b\right\|_{\mathcal{L}(\mathfrak Q_n,\mathfrak Q_m')} \;
\left\|(H_q^0+1)^{\frac{1}{2}}\otimes 1^{\otimes(m-q)}\Phi^{(m)}\right\| \,
\left\|(H_p^0+1)^{\frac{1}{2}}\otimes 1^{\otimes(n-p)}\Psi^{(n)}\right\|\,.
\end{eqnarray*}
Using the symmetry of the  vectors $\Phi^{(m)}$ (resp.~$\Psi^{(n)}$), we remark
$$
\left\|(H_q^0+1)^{\frac{1}{2}}\otimes 1^{\otimes(m-q)}\Phi^{(m)}\right\|^2=
\langle \Phi^{(m)}, (\sum_{i=1}^q A_i+1) \, \Phi^{(m)}\rangle=\langle \Phi^{(m)}, (q A_1+1) \, \Phi^{(m)}\rangle\,.
$$
(iv) Thanks to a polarization formula the monomial $c_m$ determines uniquely the operator $\tilde c_m\in\mathcal{L}(\mathfrak{Q}_{p},\mathfrak{Q}_{q}')$. In fact for any $\Phi^{(q)}\in\vee^{alg,q}Q(A)$ and $\Psi^{(p)}\in \vee^{alg,q}Q(A)$  the quantity $\langle \Phi^{(q)}, \tilde c_m \,\Psi^{(p)}\rangle$ can be written as a linear combination of $(c_m(z_i))_{i\in I}$ where
$I$ is a finite set and $z_i$ are given points in $Q(A)$. Therefore, for any $\Phi^{(q)}\in\vee^{alg,q}Q(A)$ and $\Psi^{(p)}\in \vee^{alg,p}Q(A)$ the sequence $(\langle \Phi^{(q)}, \tilde c_m \,\Psi^{(p)}\rangle)_{m\in\N}$ is convergent.
 Since $(||\tilde c_m||_{\mathcal{L}(\mathfrak{Q}_p,\mathfrak{Q}_q')})_{m\in\N}$  is bounded, one can prove  by  an $\eta/3$-argument that $\tilde c_m$ converges weakly to an operator $\tilde c\in \mathcal{L}(\mathfrak{Q}_{p},\mathfrak{Q}_{q}')$, i.e.:
 \begin{equation}
 \label{convsy}
 \langle \Phi^{(q)}, \tilde c_m \,\Psi^{(p)}\rangle \underset{{m\to\infty}}{\to} \langle \Phi^{(q)}, \tilde c \,\Psi^{(p)}\rangle\,,\quad  \forall \Phi^{(q)}\in \mathfrak{Q}_q,\forall\Psi^{(p)}\in \mathfrak{Q}_p\,.
 \end{equation}
Hence, $c(z)=\langle z^{\otimes q}, \tilde c z^{\otimes p}\rangle$ and belongs to
$\mathcal Q_{p,q}(A)$. As a consequence of \eqref{convsy}, the operator
$\tilde c_m\otimes 1^{(n-p)}$ converges also weakly to $\tilde c\otimes 1^{(n-p)}$ in
$\mathcal{L}(\mathfrak{Q}_{n},\mathfrak{Q}_{n-p+q}')$ and the convergence of
$c_m^{Wick}$ towards $c^{Wick}$ follows.\\
(v)  The relation \eqref{wicktrans} is already proved in \cite[Proposition 2.10]{AmNi1} for
symbols $b\in \mathcal P_{p,q}(\mathcal Z_0)$. In order to extend it to the class
$\mathcal Q_{p,q}(A)$ it is enough to use the approximation argument provided by (iv). Let  $\chi \in \mathcal C_{0}^{\infty}(\R)$ such that $\chi(x)=1$ if $\|x\|\leq{1},$ $\chi(x)=0$ if $\|x\|\geq{2}$ and $0\leq{\chi}\leq{1}$. We denote for $m \in \N$, $\chi_{m}(x)=\chi(\frac{x}{m})$. Let $b\in \mathcal Q_{p,q}(A)$ and consider the sequence of symbols
$$
c_m(z)= \langle z^{\otimes q}, \chi_m(H_q^0) \,\tilde b\, \chi_m(H_p^0)\, z^{\otimes p}\rangle
\in \mathcal P_{p,q}(\mathcal Z_0)\subset \mathcal Q_{p,q}(A)\,.
$$
The use of \cite[Proposition 2.10]{AmNi1} yields for any $\Phi^{(n_1)}\in\mathfrak{Q}_{n_1}$ and $\Psi^{(n_2)}\in\mathfrak{Q}_{n_2}$,
\begin{equation}
\label{eq.aptrans}
\langle \Phi^{(n_1)}, c_m^{Wick} \mathcal W(\frac{\sqrt{2}}{i\varepsilon} \xi ) \Psi^{(n_2)}\rangle
=\langle \Phi^{(n_1)}, \mathcal W(\frac{\sqrt{2}}{i\varepsilon} \xi ) \, c_m(z + \xi)^{Wick} \, \Psi^{(n_2)}\rangle\,.
\end{equation}
Now, it is easy to check that
$$
c_m\overset{b}{\to}b\,, \quad \text{ in } \mathcal Q_{p,q}(A).
$$
Moreover $c_m(\cdot+\xi)\in \oplus^{alg}_{k,l\geq 0}\mathcal P_{k,l}(\mathcal Z_0)$,
\begin{eqnarray*}
c_m(z+\xi)&=&\langle (z+\xi)^{\otimes q}, \tilde c_m\,(z+\xi)^{\otimes p}\rangle\\
&=& \displaystyle\sum_{\underset{0\leq j\leq p}{0\leq i\leq q}}
C_q^i C_p^j \, \langle z^{\otimes (q-i)}\otimes \xi^{\otimes i}, \mathcal S_{q}\,\tilde c_m\,\mathcal S_p \, z^{\otimes (p-j)}\otimes \xi^{\otimes j}\rangle\\
&=:& \displaystyle\sum_{\underset{0\leq j\leq p}{0\leq i\leq q}}
C_q^i C_p^j \, c_m^{(i,j)}(z)\,.
\end{eqnarray*}
So, it is clear that each monomial $ c_m^{(i,j)}$ in the above sum $b$-converges to
$ b^{(i,j)}=\langle z^{\otimes (q-i)}\otimes \xi^{\otimes i}, \mathcal S_{q}\,\tilde b\,\mathcal S_p \, z^{\otimes (p-j)}\otimes \xi^{\otimes j}\rangle$ since $\tilde c_m$ converges weakly to $\tilde b$ in $\mathcal{L}(\mathfrak{Q}_{p},\mathfrak{Q}_{q}')$. Remark also that  Proposition \ref{pr.invsew} shows for any $r\in\N$ that the $r^{th}$ components of the following coherent vectors satisfy
$$
\big[\mathcal W(\frac{\sqrt{2}}{i\varepsilon} \xi ) \Psi^{(n_2)}\big]^{(r)}\in \mathfrak{Q}_{r}\,\quad \text{ and } \quad
 \big[\mathcal W(\frac{\sqrt{2}}{i\varepsilon} \xi )^* \Phi^{(n_1)}\big]^{(r)}\in \mathfrak{Q}_{r}\,.
$$
Therefore using (iv) and taking the limit $m\to\infty$ in \eqref{eq.aptrans} proves the claimed identity.
\end{proof}

\bigskip
\noindent
\textit{A regularity property of Weyl operators}:
It is convenient to recall the following regularity property for the Weyl operators. Remember that the operator $d\Gamma(A)+\mathbf{N}$ is non-negative and self-adjoint on the symmetric Fock space satisfying
$$
d\Gamma(A)+\mathbf{N}_{|\vee^N\mathcal Z_0}=H_N^0+1\,, \quad \text{ when } \varepsilon=\frac{1}{N}\,.
$$
 Moreover, $d\Gamma(A)+\mathbf{N}$  has  an invariant form domain with respect to the Weyl operator $\mathcal W(\xi)$ when $\xi\in Q(A)$. This propriety can be proved using the Faris-Lavine argument \cite{FL} and it is  proved for instance in \cite{AmBr}.
\begin{prop}
\label{pr.invsew}
For any $\xi \in Q(A)$ the form domain $Q(d\Gamma(A)+\mathbf{N})$ is invariant with respect to the Weyl operator $\mathcal W(\xi)$. Moreover,  there exists uniformly in $\varepsilon\in(0,\bar \varepsilon)$  a constant $C:=C(\xi)>0$ such that
\begin{equation}
\label{eq.sewe}
\|(d\Gamma(A)+\mathbf{N})^{\frac{1}{2}}\mathcal W(
\xi)(d\Gamma(A)+\mathbf{N}+1)^{-\frac{1}{2}}\|_{\Gamma_s(\mathcal Z_{0})}\leq C\,,
\end{equation}
and in particular for any $\Psi^{(N)}\in Q(H_N^0)$, $\varepsilon=\frac{1}{N}$,
$$
\bigg\|(H_{N-1}^0+1)^{1/2}[\mathcal W(\xi)\Psi^{(N)}]^{(N-1)}\bigg\|\leq C \bigg\|
(H_N^0+1)^{1/2}\Psi^{(N)} \bigg\|\,,
$$
where $[\mathcal W(\xi)\Psi^{(N)}]^{(N-1)}$ denotes the $(N-1)^{th}$ component of $\mathcal W(\xi)\Psi^{(N)}\in
\Gamma_s(\mathcal Z_0)$.
\end{prop}

\subsection{Relationship with Wigner measures}
\label{sub.app}
Wigner measures are defined through Weyl operators nevertheless it is important  for the mean-field problem to draw the link with Wick quantization. Their relationship is clarified by the following Proposition proved in \cite[Theorem 6.13]{AmNi1} and \cite[Corollary 6.14]{AmNi1}.

\begin{prop}
\label{eq.wigcompact}
Let $\{|\Psi^{(N)}\rangle \langle \Psi^{(N)}|\}_{N \in \N}$ be a sequence of normal states on $\vee^{N}\mathcal Z_{0}$ satisfying:
\begin{equation*}
\exists  C>0, \forall N\in\N, \;\langle \Psi^{(N)},H_N^0\Psi^{(N)}\rangle \leq{CN}\,,
\end{equation*}
and
$$
\mathcal{M}(|\Psi^{(N)}\rangle \langle \Psi^{(N)}|,N \in \N)=\{ \mu \}.
$$
 Then, for any $ b \in \oplus_{p,q\geq 0}^{alg}\mathcal
P^{\infty}_{p,q}(\mathcal Z_{0})$,
\begin{eqnarray*}
&&\lim_{\underset{\varepsilon N=1}{N \to +\infty}}\langle \Psi^{(N)},b^{Wick}\Psi^{(N)}\rangle=\int_{\mathcal Z_{0}}b(z)\;d\mu(z)\,,\\
&&  \lim_{\underset{\varepsilon N=1}{N \to +\infty}}\langle \Psi^{(N)}, \mathcal W(\xi)\, b^{Wick}\Psi^{(N)}\rangle=\int_{\mathcal Z_{0}} e^{i\Re\langle z, \xi\rangle} \,b(z)\;d\mu(z)\;,
\end{eqnarray*}
\end{prop}

The following a priori estimate is a consequence of \cite[Proposition 3.11]{AmNi4},
\cite[Lemma 3.13]{AmNi4}, \cite[Lemma 2.14]{AmNi3} and \cite[Lemma 3.12]{AmNi4}.
\begin{prop}
\label{pr.transportapriori}
Let $\{|\Psi^{(N)}\rangle \langle \Psi^{(N)}|\}_{N\in \N}$ a sequence  of normal states  on $\vee^{n}\mathcal Z_{0}$ satisfying:
\begin{equation*}
\exists  C>0, \forall N\in\N, \;\langle \Psi^{(N)},H_N^0\Psi^{(N)}\rangle \leq{CN}\,,
\end{equation*}
and
$$
\mathcal{M}(|\Psi^{(N)}\rangle \langle \Psi^{(N)}|,N \in \N)=\{ \mu \}.
$$
Then the Wigner measure $\mu$ is carried by $Q(A)$ (i.e.: $\mu(Q(A))=1$) and its restriction to $Q(A)$ is a Borel probability measure on $(Q(A),||\cdot||_{Q(A)})$ fulfilling
\begin{eqnarray*}
&&\int_{\mathcal{Z}_{0}} \|z\|^{2}_{Q(A)}\,d\mu(z)
\leq C\,, \\
\text{and } && \mu(B_{\mathcal Z_0}(0,1))=1\,,
\end{eqnarray*}
where $B_{\mathcal Z_0}(0,1)$ is the unit ball of $\mathcal Z_0$.
\end{prop}

Some kind of a Fatou's lemma for Wigner measures holds true.
\begin{prop}
\label{eq.wigorder}
Let $\{|\Psi^{(N)}\rangle \langle \Psi^{(N)}|\}_{N \in \N}$ be a sequence of normal states on $\vee^{N}\mathcal Z_{0}$ satisfying:
\begin{equation*}
\exists  C>0, \forall N\in\N, \;\langle \Psi^{(N)},H_N^0\Psi^{(N)}\rangle \leq{CN}\,,
\end{equation*}
and
$$
\mathcal{M}(|\Psi^{(N)}\rangle \langle \Psi^{(N)}|,N \in \N)=\{ \mu \}.
$$
Then  for any $b\in\mathcal Q_{p,p}(A)$ such that  $\tilde b\geq 0$,
$$
\liminf_{\underset{\varepsilon N=1}{N \to +\infty}}\; \langle \Psi^{(N)},b^{Wick}\Psi^{(N)}\rangle\geq
\int_{\mathcal Z_{0}}b(z)\;d\mu(z)\,.
$$
\end{prop}
\begin{proof}
Since $b\in\mathcal Q_{p,p}(A)$,  $b(z)=\langle z^{\otimes p}, \tilde b \, z^{\otimes p}\rangle $ with $\tilde b\in \mathcal{L}(\mathfrak{Q}_p,\mathfrak{Q}_p')$, $\tilde b\geq 0$, then the quadratic form
$$
(\Psi,\Phi)\in Q(H_p^0)\times Q(H_p^0)\to \langle \Psi, \tilde b \,\Phi\rangle\,,
$$
is closed and non-negative. Hence by \cite[Theorem VIII]{RS1} there exists a unique self-adjoint operator on $\vee^p\mathcal Z_0$, denoted by $B$, such that
$\langle \Psi, \tilde b \,\Phi\rangle=\langle \Psi, B \,\Phi\rangle$ for any $\Psi,\Phi\in
D(B)$ and $D(B)$ is dense in $Q(H_p^0)$. Moreover, the inequality $0\leq B\leq c \,H_p^0$ holds  in the sense of quadratic forms on $Q(H_p^0)\subset Q(B)$. So, when $\varepsilon=\frac{1}{N}$,
$$
\langle \Psi^{(N)},  b^{Wick} \Psi^{(N)}\rangle
=\frac{N!}{N^p (N-p)!}\langle \Psi^{(N)}, B\otimes 1^{(N-p)} \Psi^{(N)}\rangle \geq
 \frac{N!}{N^p (N-p)!} \langle\Psi^{(N)}, \chi_m(B)B\otimes 1^{(N-p)} \Psi^{(N)}\rangle\,,
$$
where $\chi_m$ is a suitable cutoff function such that $0\leq \chi_m \leq 1$ and $\chi_m\to 1$ when $m\to\infty$. For any compact operator $C$ on $\vee^p\mathcal Z_0$ satisfying $0\leq C \leq \chi_m(B)B$, one get
$$
\langle \Psi^{(N)},  b^{Wick} \Psi^{(N)}\rangle \geq \frac{N!}{N^p (N-p)!}
\langle \Psi^{(N)}, C\otimes 1^{(N-p)} \Psi^{(N)}\rangle\,.
$$
So using Proposition \ref{eq.wigcompact} one obtains
$$
\liminf_{\underset{\varepsilon N=1}{N\to\infty}}\langle \Psi^{(N)},  b^{Wick} \Psi^{(N)}\rangle \geq \int_{\mathcal Z_0} \langle z^{\otimes p}, C \, z^{\otimes p}\rangle\,d\mu,
$$
for any  non-negative compact operator $C$ such that $C\leq \chi_m(B) B$. Remark that there exists a sequence of such operators $C_k$ which converges strongly to $\chi_m(B) B$.
Therefore using Proposition  \ref{pr.transportapriori} and dominated convergence one obtains
 $$
\liminf_{\underset{\varepsilon N=1}{N\to\infty}}\langle \Psi^{(N)},  b^{Wick} \Psi^{(N)}\rangle \geq \int_{\mathcal Z_0} \langle z^{\otimes p}, B \, z^{\otimes p}\rangle\,d\mu=
\int_{Q(A)} b(z)\,d\mu\,.
$$
\end{proof}
\section{Measure valued solutions to continuity equation}
 In this appendix we recall the key result from \cite{W1} which provides uniqueness of solutions satisfying the Liouville equation  \eqref{liouv1}-\eqref{liouv2}. We adapt  
 \cite[Theorem 2.4]{W1} to the framework of Section \ref{se.prem-resut}.  
 Remember  that $A$ is a non-negative self-adjoint operator with $Q(A)$ its form domain and $Q'(A)$ its dual. We consider the following Liouville's equation on $\R$,
\begin{equation}
\label{eq.transport}
      \displaystyle\int_{\R}\int_{Q'(A)}\partial_{t}\varphi(t,x)+\Re \langle
  v_t(x),\nabla\varphi(t,x)\rangle_{Q'(A)} \; d\mu_{t}(x)\,dt=0, \quad \forall
    \varphi \in \mathcal{C}_{0,cyl}^{\infty}(\R\times Q'(A))\,,
\end{equation}
where $\mu_t$ belongs to $\mathfrak{P}(Q(A))$ and $v_t$ is the vector field, 
$$v_{t}(z):=-ie^{itA}\partial_{\bar z}q_{0}(e^{-itA}z): \R \times Q(A) \to \mathcal Z_{0},$$
with $\partial_{\bar z}q_{0}$ defined in \eqref{q0}. We shall also consider  the following Cauchy problem in $Q(A)$,
\begin{equation}
\label{eq.cauchyapp}
\partial_{t}\gamma(t)=v_{t}(\gamma(t)),\quad \gamma(0)=z \in Q(A)\,.
\end{equation}
Observe that \eqref{eq.cauchyapp} is equivalent to the mean-field equation \eqref{field-eq2} in the sense that  $\gamma\in\mathcal C(I,Q(A)) \cap \mathcal C^{1}(I,Q'(A))$ is a strong solution of \eqref{field-eq2} if and only if $\tilde \gamma(t)=e^{it A}\gamma(t)$ is a strong solution of 
 \eqref{eq.cauchyapp}. In particular, assumption \eqref{C1} provides the local well-posedness of the Cauchy problem \eqref{eq.cauchyapp} and we have a local flow denoted by $\tilde\Phi(t,0)$ and given as  $$\tilde\Phi(t,0)= e^{itA}\circ\Phi(t,0),$$ with $\Phi(t,0)$ the flow of the mean-field equation defined in \eqref{flow}.

\begin{thm}
\label{pr.D5}
Consider $A$ and $q_0$ as in Section \ref{se.prem-resut} and assume \eqref{A1} and \eqref{C1} are true. Let $t\in \R \to\tilde\mu_{t}\in \mathfrak{P}(Q(A))$ be a weakly narrowly continuous solution in $\mathfrak{P}(Q'(A))$ of the Liouville equation \eqref{eq.transport}. Let $I$ be the open interval  provided by \eqref{C1} and assume additionally that:
\begin{enumerate}
  \item [(i)] There exists  $C>0$ such that $ \displaystyle\int_I \int_{Q(A)} ||x||^2_{Q(A)} d\tilde\mu_t(x)dt\leq C$.
  \item [(ii)]  For all $t\in I$, $\tilde\mu_t(B_{\mathcal Z_0}(0,1))=1$.
\end{enumerate}
Then $\tilde\mu_t=\tilde\Phi(t,0)_\sharp\mu_s$ for all $t\in I$ with $\tilde\Phi(t,0)$ is the local flow of the initial value problem \eqref{eq.cauchyapp}.
\end{thm}
 \subsection*{Acknowledgement}
I would like to thank Zied Ammari and Francis Nier for their precious help and support during the preparation of this paper.

\end{document}